\documentclass[acmsmall]{acmart}

\usepackage{caption}
\usepackage{amsmath,amsfonts}
\usepackage{amsthm}
\usepackage{algorithmic}
\usepackage{graphicx}
\usepackage{textcomp}
\usepackage{subfigure}
\usepackage{booktabs}
\usepackage{xcolor}
\usepackage{multirow}
\usepackage{array}
\usepackage{threeparttable}
\usepackage[ruled,vlined]{algorithm2e}
\usepackage{balance}
\usepackage{lipsum}
\usepackage{csquotes}
\usepackage{makecell}

\usepackage{amssymb}
\usepackage[misc]{ifsym} 
\usepackage{tikz}
\usetikzlibrary{patterns,decorations.pathreplacing,positioning,shadows,calc,3d}

\theoremstyle{definition}
\newtheorem{definition}{Definition}

\theoremstyle{lemma}
\newtheorem{lemma}{Lemma}

\theoremstyle{theorem}
\newtheorem{theorem}{Theorem}

\AtBeginDocument{%
  }

\setcopyright{cc}
\setcctype{by}
\acmJournal{PACMMOD}
\acmYear{2026} \acmVolume{4} \acmNumber{3 (SIGMOD)} \acmArticle{241}
\acmMonth{6} \acmDOI{10.1145/3802118}

\begin{document}

\title{TaCo: Data-adaptive and Query-aware Subspace Collision for High-dimensional Approximate Nearest Neighbor Search}


\author{Jiuqi Wei}
\affiliation{%
  \institution{Oceanbase, Ant Group}
  \country{China}
}
\email{weijiuqi.wjq@antgroup.com}

\author{Zhenyu Liao}
\affiliation{%
  \institution{Huazhong University of Science and Technology}
  \country{China}
}
\email{zhenyu_liao@hust.edu.cn}

\author{Ruoyu Han}
\affiliation{%
  \institution{Institute of Computing Technology, Chinese Academy of Sciences}
  \country{China}
}
\email{hanruoyu23z@ict.ac.cn}

\author{Quanqing Xu}
\affiliation{%
  \institution{Oceanbase, Ant Group}
  \country{China}
}
\email{xuquanqing.xqq@oceanbase.com}

\author{Chuanhui Yang}
\affiliation{%
  \institution{Oceanbase, Ant Group}
  \country{China}
}
\email{rizhao.ych@oceanbase.com}
\thanks{$^\ast$Chuanhui Yang is the corresponding author.}

\author{Themis Palpanas}
\affiliation{%
  \institution{LIPADE, Universit{\'e} Paris Cit{\'e}}
  \country{France}
}
\email{themis@mi.parisdescartes.fr} 


\renewcommand{\shortauthors}{Jiuqi Wei et al.}
%% No italics, no superscripts
%% If needed use a foot or author note to identify equal contribution

\begin{abstract}
Approximate Nearest Neighbor Search (ANNS) in high-dimensional Euclidean spaces is a fundamental problem with broad applications.
\emph{Subspace Collision} is a newly proposed ANNS framework that provides a novel paradigm for similarity search and achieves superior indexing and query performance.
However, the subspace collision framework remains data-agnostic and query-oblivious, resulting in imbalanced index construction and wasted query overhead.
In this paper, we address these limitations from two aspects: first, we design a subspace-oriented data transformation mechanism by averaging the entropies computed over each subspace of the transformed data, which ensures balanced subspace partitioning (in an information theoretical sense) and enables \emph{data-adaptive} subspace collision; second, we present \emph{query-aware} and scalable query strategies that dynamically allocate overhead for each query and accelerate collision probing within subspaces.
Building on these ideas, we propose a novel data-adap\textbf{\underline{t}}ive and query-\textbf{\underline{a}}ware subspa\textbf{\underline{c}}e c\textbf{\underline{o}}llision method, abbreviated as TaCo, which achieves efficient and accurate ANN search while maintaining an excellent balance between indexing and query performance.
Extensive experiments on real-world datasets demonstrate that, when compared to state-of-the-art subspace collision methods, TaCo achieves up to $\mathbf{8\times}$ speedup in indexing and reduces to $\mathbf{0.6\times}$ memory footprint, while achieving over $\mathbf{1.5\times}$ query throughput.
Moreover, TaCo achieves state-of-the-art indexing performance and provides an effective balance between indexing and query efficiency, even when compared with advanced methods beyond the subspace-collision paradigm.
\end{abstract}




\received{October 2025}
\received[revised]{January 2026}
\received[accepted]{February 2026}

\maketitle

% \blfootnote{$^\ast$ Chuanhui Yang is the corresponding author.}

\section{Introduction}
\textbf{Background and Problem.} 
Nearest Neighbor Search (NNS) on high-dimensional vectors is a fundamental problem for numerous applications, such as recommendation systems~\cite{schafer2007collaborative}, information retrieval~\cite{karpukhin2020dense}, and data mining~\cite{tagami2017annexml}.
However, NNS in high-dimensional spaces is challenging due to the \emph{curse of dimensionality} phenomenon~\cite{hinneburg2000nearest, borodin1999lower,louart2018random,liao2021random,couillet2022RMT4ML,liao2025Randoma}. 
In practice, Approximate Nearest Neighbor Search (ANNS) is often used as an alternative approach, achieving huge efficiency gains by sacrificing some query accuracy~\cite{zeyubulletin-sep23,li2019approximate,annbulletin}. 
With the development of Large Language Models (LLMs)~\cite{zhao2023survey}, ANNS has received a lot of attention. 
For example, Retrieval Augmented Generation (RAG) leverages ANNS to add context to an LLM query~\cite{lewis2020retrieval,fan2024survey,gao2023retrieval}, and Key-Value Cache (KVCache) adopts ANNS to accelerate the model inference process~\cite{zhang2025pqcache,liu2024retrievalattentionacceleratinglongcontextllm,desaihashattention}.
Depending on the storage medium of indexes and datasets, ANNS methods can be divided into in-memory methods~\cite{malkov2018efficient,detlsh,ge2013optimized,gao2024rabitq} and memory-disk hybrid methods~\cite{jayaram2019diskann,chen2021spann,wang2024starling}.
In this paper, we focus on in-memory methods due to their widespread application.

\textbf{Prior Work.}
Numerous excellent ANNS methods have been developed, falling into four categories:
locality-sensitive hashing (LSH)-based methods~\cite{dblsh,lccslsh,pmlsh,andoni2015optimal,detlsh,pdetlsh}, vector quantization (VQ)-based methods ~\cite{jegou2010product,ge2013optimized,norouzi2013cartesian,babenko2014inverted,gao2024rabitq,gao2025practical},
tree-based methods ~\cite{annoy,coconut,messi,dumpy,seanetconf,leafi}, and graph-based methods ~\cite{fu2019fast,malkov2018efficient,voruganti2025mirage,lshapg,azizi2023elpis,gou2025symphonyqg}.
Although each category of methods exhibits distinct strengths, they commonly face inherent limitations that restrict their ability to simultaneously perform well in index construction and query answering~\cite{li2019approximate}.
\emph{Subspace Collision} is a newly proposed ANNS framework~\cite{wei2025subspace}.
% Unlike traditional methods (LSH-, VQ-, tree-, and graph-based), subspace collision works in a completely new manner.
Within the subspace collision framework, the original high-dimensional space is first partitioned into $N_s$ subspaces. 
Subsequently, collisions between the query point and data points are counted within each subspace. 
A data point is considered to \emph{collide} with a query in a subspace if it ranks among the points with the closest distance (within that subspace) to that query. %(predefined fraction) data points to that query.
% A \emph{collision} is defined as the preset proportion of data points closest to the query point within a subspace.
For any data point, its number of collisions in all subspaces is defined as the \emph{SC-score} metric, an integer value ranging from 0 to $N_s$. 
Empirical analysis in various datasets shows that the SC-score metric adheres to the \emph{Pareto principle} (also known as the 80-20 rule), which makes it an efficient proxy for the Euclidean distance between the query point and data points.
\emph{SuCo}~\cite{wei2025subspace} is the first approach designed for the subspace collision framework.
The corresponding experiments demonstrate that SuCo achieves state-of-the-art performance.
% LSH-based methods provide robust theoretical guarantees regarding result quality, but lead to significant query latency and high memory overhead~\cite{detlsh,dblsh}.
% VQ-based methods employ quantization codebooks to compress data, achieving significant memory reduction. However, their lossy compression typically compromises the search accuracy, thus requiring a lot of indexing time to obtain fine-grained indexes ~\cite{ge2013optimized,gao2024rabitq}.
% Tree-based methods partition data points through recursive node splitting based on data locality, yet their effectiveness degrades as the space dimensionality increases, thereby constraining query performance~\cite{bohm2000cost, weber1998quantitative}.
% Graph-based methods typically achieve superior query performance, but require significantly more time to construct indexes and consume greater memory~\cite{hydra2,li2019approximate}.
% This is because during index building, they have to find appropriate neighbors for each data point and connect them, which is a computationally intensive process.
% In contrast, query processing simply follows a converging path through these connections, enabling efficient search.
% Consequently, it is difficult for traditional ANNS methods to perform well in both index construction and query answering, while providing theoretical guarantees for answer quality.

\textbf{Limitations and Motivation.}
Although the subspace collision framework fits the ANNS tasks well and SuCo shows superior performance, we  identify some limitations that need to be fixed. 
(1) The subspace collision framework employs a naive, \emph{data-agnostic} subspace partitioning mechanism: uniformly dividing the original $d$-dimensional space into $N_s$ fixed-size subspaces of $s=\lfloor \frac{d}{N_s} \rfloor$ dimensions. 
This mechanism disregards inherent data distributions, resulting in subspaces with imbalanced statistical properties. 
However, the SC-score metric still assigns equal weight to all subspaces during computation despite this imbalance.
Consequently, designing a \emph{data-adaptive} subspace partitioning mechanism becomes essential to mitigate errors arising from data skew across subspaces.
(2) The query strategies of the subspace collision framework are \emph{query-oblivious}, employing the same query overhead for all queries.
However, empirical studies confirm significant performance variance across queries in ANNS tasks~\cite{li2020improving,wang2024boldsymbol}.
This rule also applies to the subspace collision framework.
For instance, some queries demonstrate highly discriminative SC-score distributions, thus high recall can be achieved by re-ranking only the top-scoring data points. 
In contrast, other queries exhibit poor SC-score discriminability, requiring larger candidate sets during re-ranking to achieve comparable recall. 
Consequently, it is necessary to design \emph{query-aware} query answering strategies for the subspace collision framework to dynamically adjust to the query overhead.
(3) Although SuCo's inverted multi-index (IMI) demonstrates competitive performance, the dynamic activation algorithm designed for IMI is not scalable enough.
By maintaining the activation list via a linear array, the \emph{Dynamic Activation} algorithm~\cite{wei2025subspace} incurs linear query complexity.
% While this design performs satisfactorily with a small length of the IMI list, it becomes a bottleneck when an increased list length is required for higher indexing precision.
% The linear query complexity constrains overall efficiency under such conditions.
While the linear design performs satisfactorily with a small length of the IMI list, it becomes an efficiency bottleneck when a large list length is required for higher indexing precision.
Thus, designing a more \emph{scalable} algorithm for IMI index structure becomes essential.

\textbf{Our Method.}
In this paper, we propose a \emph{data-adaptive} and \emph{query-aware} subspace collision framework for ANNS and design a novel method, named TaCo, which exhibits superior performance in both index construction and query processing compared to state-of-the-art ANNS methods.
\underline{First}, we analyze the data-agnostic issue of the original subspace collision framework, which fails to account for inherent data distributions, resulting in subspaces with imbalanced statistical properties.
To address this issue, we formulate an \emph{entropy-averaging optimization problem} and show that it can be solved efficiently, which leads to the design of a \emph{subspace-oriented data transformation} mechanism. 
This transformation makes the subspace collision framework \emph{data-adaptive} and enables more accurate subspace partitioning.
Furthermore, the transformation reduces dimensionality by 40\% to 96\% %95.83\% 
across different datasets, leading to significant improvements in both indexing and query efficiency.
\underline{Second}, we analyze the query-oblivious issue of the original subspace collision framework, i.e., using the same strategy for all queries results in wasted query overhead.
To address this issue, we design a \emph{query-aware candidate selection} mechanism that leverages the distribution of SC-scores (an intermediate query state) to \emph{dynamically} adjust the query overhead for each individual query.
This mechanism makes the subspace collision framework \emph{query-aware} and improves the query efficiency.
\underline{Third}, we analyze the scalability issue of the query algorithm for inverted multi-index (IMI) structure, revealing that its performance under fine-grained indexes is constrained by the linear query complexity.
To address this issue, we design a \emph{scalable dynamic activation} algorithm that employs a min-heap structure to reduce the query complexity to $\mathcal{O}(1)$.
This advancement supports efficient querying even with highly fine-grained indexes.
\underline{Fourth}, integrating the above designs, we propose a data-adap\textbf{\underline{t}}ive and query-\textbf{\underline{a}}ware subspa\textbf{\underline{c}}e c\textbf{\underline{o}}llision method, abbreviated as TaCo.
To systematically evaluate the proposed optimization techniques, we also present three ablation methods based on SuCo, where SuCo~\cite{wei2025subspace} is the state-of-the-art subspace collision-based method.
\underline{Fifth}, we conduct extensive experiments on real-world datasets.
Empirical results show that when compared with subspace collision-based methods, TaCo achieves up to $\mathbf{8\times}$ speedups in indexing while consuming only $\mathbf{0.6\times}$ memory.
In terms of query performance, TaCo achieves over $\mathbf{1.5\times}$ query throughput at equivalent high recall. 
% Moreover, TaCo achieves top performance and maintains a strong balance between index construction and query processing, even when compared to non-subspace collision-based methods. 
Moreover, TaCo achieves state-of-the-art indexing performance and provides an effective balance between indexing and query efficiency, even when compared with advanced methods beyond the subspace-collision paradigm.

Our main contributions are summarized as follows.

\begin{itemize} 
        \item We design a subspace-oriented data transformation mechanism by formulating an entropy-averaging optimization problem and showing that it can be solved efficiently.
        This approach achieves balanced subspace partitioning with impressive dimensionality reduction, thereby enhancing the subspace collision framework with data adaptability.
        \item We present query-aware and scalable query strategies that dynamically adjusts the overhead for each individual query and accelerate the collision probing within each subspace, thereby enhancing the subspace collision framework with query awareness and scalability.
        \item We propose a data-adaptive and query-aware subspace collision method named TaCo, which supports efficient and accurate ANN search while maintaining the excellent balance between indexing and query performance inherent to subspace collision-based methods.
        \item We conduct extensive experiments demonstrating that when compared to state-of-the-art subspace collision methods, TaCo achieves up to $\mathbf{8\times}$ speedup in indexing and reduces to $\mathbf{0.6\times}$ memory footprint, while achieving over $\mathbf{1.5\times}$ query throughput.
        Moreover, TaCo achieves state-of-the-art indexing performance and provides an effective balance between indexing and query efficiency, even when compared with advanced methods beyond the subspace-collision paradigm.
\end{itemize}

\section{Preliminaries}

\subsection{Problem Definition}

\begin{definition}[Nearest Neighbor Search, NNS]
Given a dataset $\mathcal D$ of $n$ data points in $d$-dimensional Euclidean space $\mathbb{R}^d$ and a query $q \in \mathbb{R}^d$. 
An NNS returns a point $o^* \in \mathcal D$ which has the minimum Euclidean distance to $q$ among all points in $\mathcal D$, i.e., for each $o \in D$ satisfying $\left\|q,o^*\right\| \leq \left\|q,o\right\|$.
\end{definition} 

\begin{definition}[$k$-Nearest Neighbor Search, $k$-NNS]
Given a dataset $\mathcal D$ of $n$ data points in $d$-dimensional Euclidean space $\mathbb{R}^d$, a query $q \in \mathbb{R}^d$, and an integer $k$.
Let $o^*_i$ be the $i$-th exact nearest neighbor of $q$ in $\mathcal D$, and $\mathcal B=\{o^*_1,o^*_2,\ldots,o^*_k\}$.
A $k$-NNS returns $k$ points $\mathcal R=\{o_1,o_2,\ldots,o_k\}$, satisfying $\mathcal R=\mathcal B$.
\end{definition} 

\begin{definition}[$k$-Approximate Nearest Neighbor Search, $k$-ANNS]
Given a dataset $\mathcal D$ of $n$ data points in $d$-dimensional Euclidean space $\mathbb{R}^d$, a query $q \in \mathbb{R}^d$, an approximation ratio $c > 1$, and an integer $k$.
Let $o^*_i$ be the $i$-th exact nearest neighbor of $q$ in $\mathcal D$.
A $k$-ANNS returns $k$ points $o_1,o_2,\ldots,o_k$.
For each $o_i \in D$ satisfying $\left\|q,o_i\right\| \leq c \cdot \left\|q,o^*_i\right\|$, where $i \in [1,k]$.
\end{definition} 
% In practice, $k$-ANNS implementations typically do not have an explicit approximation ratio $c$ to characterize the quality of the results (except for LSH-based methods~\cite{detlsh}). 
In practice, many $k$-ANNS implementations do not explicitly rely on an approximation ratio $c$ during query processing; instead, they constrain the quality of the results through evaluation metrics~\cite{wei2025subspace,malkov2018efficient}.
Let $\mathcal R^* = \{o^*_1, o^*_2,\ldots,o^*_k\}$ be the ground truth set of exact $k$ nearest neighbors for query $q$ in dataset $\mathcal D$, and $\mathcal R= \{o_1, o_2,\ldots,o_k\}$ represent the returned result set.
$k$-ANNS seeks to maximize $recall=\frac{\left| \mathcal R \cap \mathcal R^* \right|}{k}$ to obtain higher quality results.
% Higher $recall$ directly reflects higher result quality.

\begin{table}
% \vspace*{-0.2cm}
% \small
% \footnotesize
	\centering
	\caption{Notations}
 % \vspace*{-0.3cm}
	\label{table1}
 {
	\begin{tabular}{cc}
		\toprule
		\textbf{Notation} & \textbf{Description} \\
		\midrule
		$\mathbb{R}^d$ & $d$-dimensional Euclidean space \\
		$\mathcal D$ & Dataset of points in $\mathbb{R}^d$ \\
		$n$ & Dataset cardinality  $\lvert \mathcal D \rvert$ \\
		$o, q$ & A data point in $\mathcal D$ and a query point in $\mathbb{R}^d$ \\
		$o^*_i$ & The $i$-th nearest data point to $q$ in $\mathcal D$ \\
        $o_i$ & The $i$-th data point in $\mathcal D$ \\
		$\left\|o_1, o_2\right\|$ & The Euclidean distance between $o_1$ and $o_2$ \\
        $o^{\prime},q^{\prime}$ & $o$ and $q$ in a subspace \\
        $N_s$ & Number of subspaces\\ 
        $s$ & Dimension of each subspace \\
        $S_i, \mathcal D_i$ & The $i$-th subspace and all data points in $S_i$ \\
        $o^i_j$, $q^i$ & The $j$-th data point and $q$ in the $i$-th subspace \\
        $\alpha$ & Collision ratio \\
		$\beta$ & Re-rank ratio \\
        $K, t$ & Number of K-means clusters and iterations\\
		\bottomrule
	\end{tabular}
 % \vspace*{-0.2cm}
 } % font size
\end{table}

\subsection{Subspace Collision Framework} \label{pre_subcollision}

\emph{Subspace Collision} is a newly proposed ANNS framework~\cite{wei2025subspace}.
Specifically, subspace collision is inspired by an intuition: two data points that are close in the high-dimensional original space are also more likely to be close in a common subspace. 
However, the distances between data points are not uniformly distributed across dimensions, which can cause the above intuition to fail.
To overcome this problem, the authors proposed the definitions of \emph{subspace sampling}, \emph{subspace collision}, and \emph{SC-score}.

\begin{definition}[Subspace Sampling]\label{def_subspace_sampling}
	Given a dataset $\mathcal D$ of $n$ data points in $d$-dimensional space, we adopt a multi-round sampling strategy to obtain $N_s$ subspaces.
 In round $i$, a number of $s=\lfloor \frac{d}{N_s} \rfloor$ dimensions are uniformly sampled without replacement to form a subspace $S_i$, $i=1,2,\ldots,N_s-1$. For the last subspace $S_{N_s}$, it simply pick up all remaining dimensions.
\end{definition}

\begin{definition}[Subspace Collision]\label{def_subspace_collision}
	Given a dataset $\mathcal D$ of $n$ data points in $d$-dimensional space, a query $q \in \mathbb{R}^d$, and a collision ratio $\alpha \in \left(0, 1\right)$. 
 We randomly select $s$ dimensions from all $d$ dimensions as a subspace $\mathbb{R}^s$ ($s<d$).
 The dataset $\mathcal D$, the data point $o$, and the query point $q$ in this subspace are denoted as $\mathcal D^\prime$, $o^\prime$, and $q^\prime$. If a point $o \in \mathcal D$ satisfies: $o^\prime$ is one of the $(\alpha \cdot n)$-NNs of $q^\prime$ in $\mathcal D^\prime$, we say that $o$ collides with $q$ in the subspace $\mathbb{R}^s$.
\end{definition} 

\begin{definition}[SC-score]\label{def_similarity}
	Given a dataset $\mathcal D$ of $n$ data points in $d$-dimensional space, a query $q \in \mathbb{R}^d$, $N_s$ $s$-dimensional subspaces, a collision ratio $\alpha \in \left(0, 1\right)$. Probe collisions of $q$ in $N_s$ subspaces with the collision ratio $\alpha$.
 For a data point $o \in \mathcal D$, its SC-score is the number of subspaces where it collides with $q$. Therefore, SC-score is an integer in $\left[0, N_s\right]$.
\end{definition}

\begin{figure} [tb]
% \vspace*{-0.2cm}
	% \flushleft
        \subfigcapskip=5pt
        \subfigure[DEEP1M: $N_s$=8, $\alpha$=0.05]{
		\includegraphics[width=0.32\linewidth]{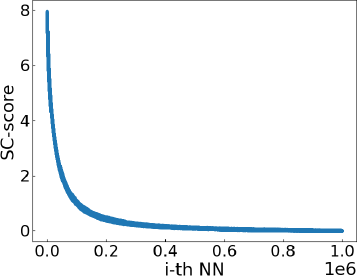}
		\label{deep_1m}}\hspace{2mm}
	\subfigure[GIST1M: $N_s$=8, $\alpha$=0.05]{
		\includegraphics[width=0.32\linewidth]{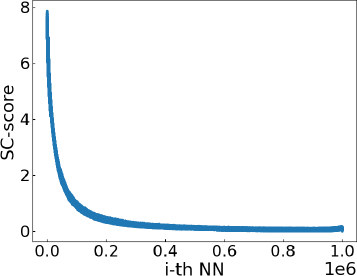}
		\label{gist1m_scscore}}\hspace{2mm}
	\subfigure[Yandex DEEP10M: $N_s$=8, $\alpha$=0.1]{
		\includegraphics[width=0.32\linewidth]{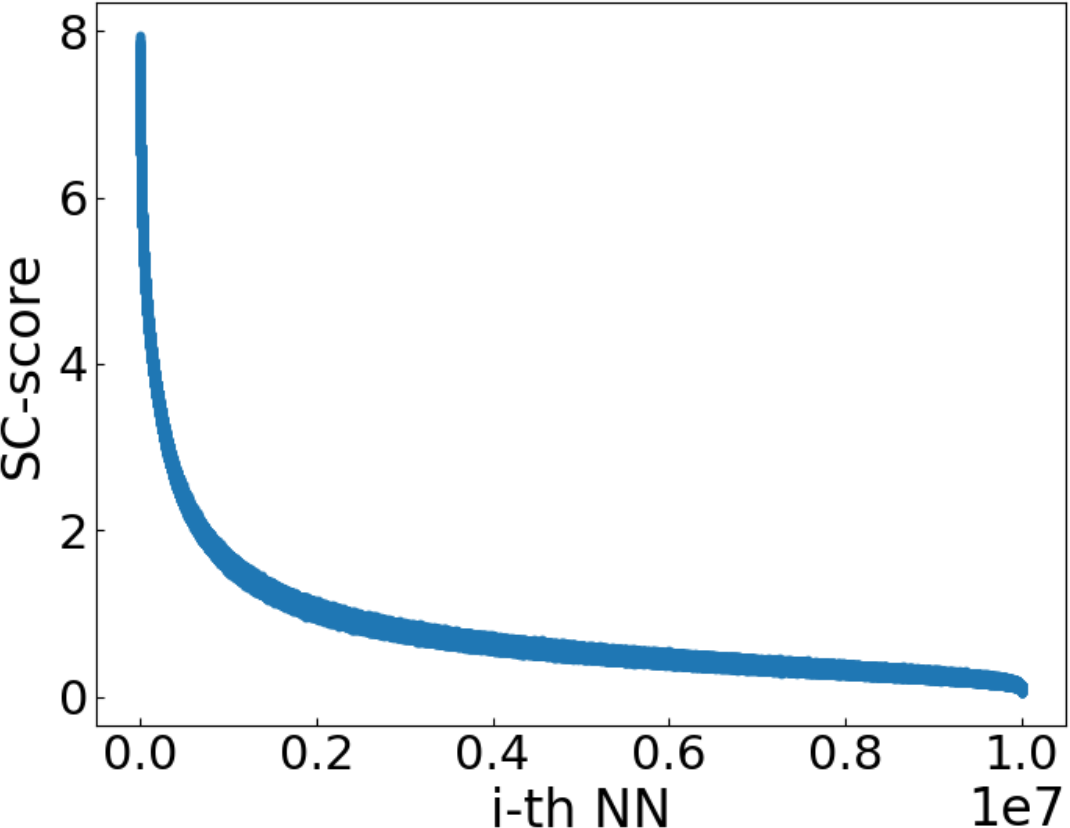}
		\label{deep10m_scscore}}\hspace{2mm}
        \subfigure[SPACEV10M: $N_s$=10, $\alpha$=0.1]{
		\includegraphics[width=0.32\linewidth]{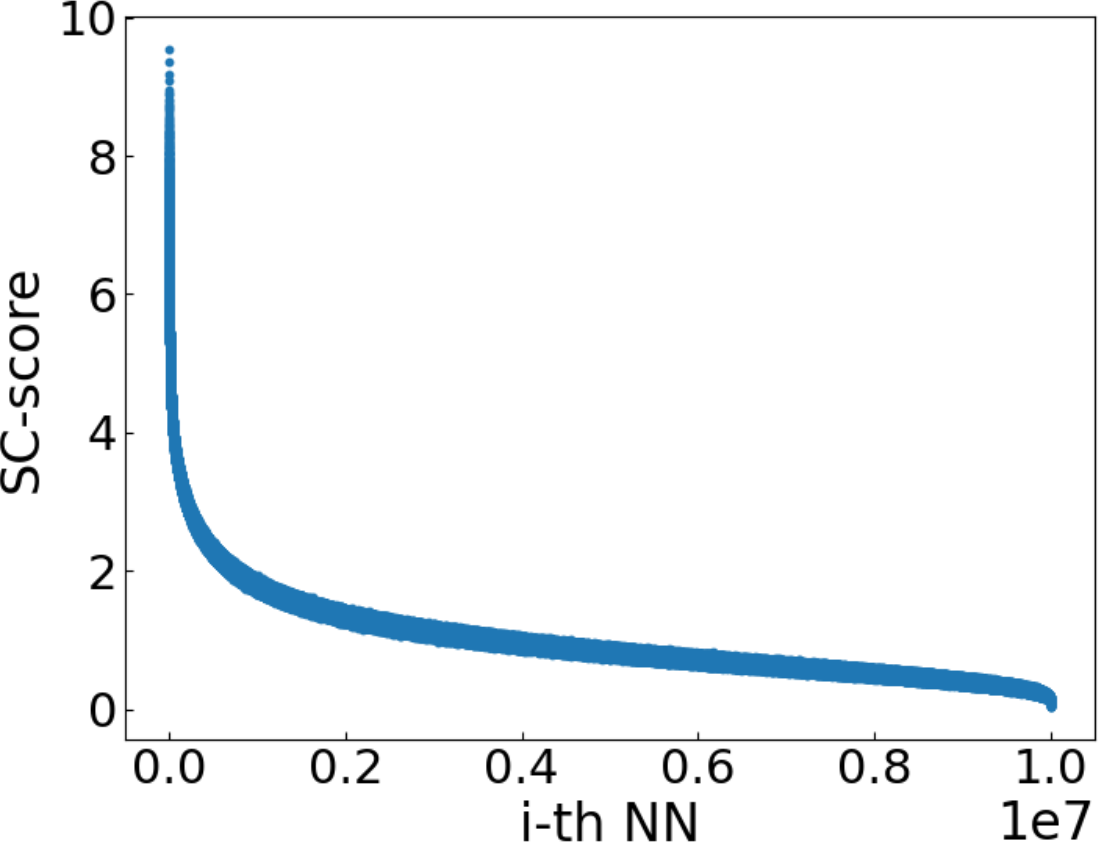}
		\label{spacev10m_scscore}}
  % \vspace*{-0.2cm}
	\caption{\enquote{Pareto principle} of SC-score.}
	\label{scscore}
 % \vspace*{-0.2cm}
\end{figure}

Subspace sampling (Def.~\ref{def_subspace_sampling}) provides a strategy for subspace partitioning. 
In practice, for convenience, the $d$-dimensional high-dimensional space is typically partitioned uniformly into $N_s$ subspaces with equal dimensions $s=\lfloor \frac{d}{N_s} \rfloor$.
Subspace collision (Def.~\ref{def_subspace_collision}) quantifies the similarity between data points and the query point within a single subspace, mitigating the influence of intra-subspace distance rankings on final results. 
Conversely, the SC-score (Def.~\ref{def_similarity}) integrates collision outcomes across multiple subspaces to alleviate significant errors potentially introduced by any single subspace.

An exciting observation is that SC-score follows the \enquote{Pareto principle} (also known as the 80-20 rule), which makes it an efficient proxy for the Euclidean distance. 
Specifically, as shown in Figure~\ref{scscore}, the top 20\% of data points nearest to the query exhibit discriminatively high SC-scores, while the remaining 80\% of the data points have low SC-scores with poor discriminability.
This phenomenon enables us to selectively process high-scoring data points to obtain high-recall ANNS results with minimal computational overhead.

% \begin{algorithm}[tb]
% % \small
% % \footnotesize
% 	\caption{SC-Linear}                                                                           
% 	\label{SC-Linear}
% 	\LinesNumbered
% 	\KwIn{Dataset $\mathcal D$, dataset size $n$, data dimensionality $d$, query point $q$, number of results $k$, subspace number $N_s$, collision ratio $\alpha$, re-rank ratio $\beta$}
% 	\KwOut{$k$ nearest points to $q$ in $\mathcal D$}
%         Initialize an array $SC\_scores$ of length $n$ and set all elements to 0; \\
% 	Divide the $d$-dimensional space into $N_s$ subspaces: $S_1, S_2,\ldots, S_{N_s}$; \\
%         Divide all data points $o_1,\ldots,o_n$ and $q$ into $N_s$ subspaces; \\
% 	\For{$i=1$ to $N_s$}{
%             \For{$j=1$ to $n$}{
%                 Calculate the Euclidean distance between $o^i_j$ and $q^i$; \\
%             }
%             Sort $o^i_1,\ldots,o^i_n$ by their distances to $q^i$ in $S_i$; \\
%             \For{$z=1$ to $\alpha \cdot n$}{
%                 Select the $z$-th-nearest point to $q^i$ whose id in $\mathcal D$ is $t$; \\
%                 $SC\_scores[t]$++; \\
%             }
% 	}
%         Sort $SC\_scores$ in descending order; \\
%         \For{$z=1$ to $\beta \cdot n$}{
%             Select the point with the $z$-th-largest SC-score in $SC\_scores$ whose id in $\mathcal D$ is $t$; \\
%             Calculate the Euclidean distance between $o_t$ and $q$; \\
%         }
% 	\Return the \emph{top}-$k$ points closest to $q$ in the $\beta \cdot n$ candidates; \\
% \end{algorithm}

\subsection{SC-Linear and SuCo Methods}

% SC-Linear~\cite{wei2025subspace} is a baseline~ANNS method (without index structure) designed for the subspace collision framework.
% As shown in Algorithm~\ref{SC-Linear}, in addition to variable initialization and subspace partitioning (lines~1-3), the core logic of the algorithm comprises three sequential phases: 
% (1) Collision counting: for each subspace, probing colliding data points with the query based on actual Euclidean distances (lines~4-10); 
% (2) Candidate selection: select a set of high-quality candidates by sorting data points according to SC-score in descending order (lines~11-13);
% (3) Result refinement: re-rank all candidate points in the original space and return the top-$k$ results (lines~14-15).
% Experimental results demonstrate SC-Linear achieves \textbf{>0.99} recall on benchmark datasets, validating the subspace collision framework's efficacy for high-accuracy ANNS.

SC-Linear~\cite{wei2025subspace} is a baseline~ANNS method (without index structure) designed for the subspace collision framework.
The core logic of the algorithm comprises three sequential phases: 
(1) Collision counting: for each subspace, probing colliding data points with the query based on actual Euclidean distances; 
(2) Candidate selection: select a set of high-quality candidates by sorting data points according to SC-score in descending order;
(3) Result refinement: re-rank all candidate points in the original space and return the top-$k$ results.
Experimental results demonstrate SC-Linear achieves over 0.99 recall on benchmark datasets, validating the subspace collision framework's efficacy for high-accuracy ANNS.

SuCo~\cite{wei2025subspace} is the state-of-the-art ANNS method developed upon the subspace collision framework. 
Compared with SC-Linear, SuCo further solves the problem of \enquote{\emph{how to count collisions in each subspace as quickly and accurately as possible}}.
Specifically, SuCo adopts the \emph{inverted multi-index (IMI)}~\cite{babenko2014inverted} as the index structure, and designs the \emph{Dynamic Activation Algorithm} for IMI to achieve efficient query.
% IMI replaces the vector quantization inside the inverted index (IVF) of K-means with the product quantization~\cite{jegou2010product}.
% Therefore, IMI can reduce the complexity of clustering from $\mathcal{O}(K \cdot n \cdot d \cdot t)$ to $\mathcal{O}(\sqrt{K} \cdot n \cdot d \cdot t)$, where $K$ is the number of K-means clusters and $t$ is the number of K-means iterations.
Experimental results confirm that SuCo achieves significantly higher efficiency than SC-Linear with an acceptable sacrifice of accuracy.
SuCo also achieves superior performance in both index construction and query answering compared to state-of-the-art ANNS methods, such as HNSW~\cite{malkov2018efficient}, OPQ~\cite{ge2013optimized}, and DET-LSH~\cite{detlsh}.

\section{Data-adaptive Subspace Collision} \label{data_adaptive_mechanism}

\subsection{Problem Analysis}\label{subsec:problem_analysis}

As introduced in Section~\ref{pre_subcollision}, the subspace collision framework adopts a naive, \emph{data-agnostic} subspace partitioning mechanism that uniformly divides the original $d$-dimensional space into $N_s$ fixed-size subspaces of $s=\lfloor \frac{d}{N_s} \rfloor$ dimensions.
This mechanism ignores the inherent data distribution and assigns equal weight to all subspaces. 
However, due to the data skew across subspaces, the significance of subspaces is imbalanced in practice.
Therefore, it is necessary to design a \emph{data-adaptive} subspace partitioning mechanism to mitigate errors arising from data skew across subspaces.

Based on characteristic analysis of the subspace collision framework, two potential technical routes may support data-adaptive subspace partitioning:
(1) first partition the subspace, and then assign different weights to each subspace according to underlying data distributions; or
(2) first transform data points according to underlying data distribution, and then partition them into subspaces.

For the first technical pathway, we explored various weighting schemes derived from data dispersion measures, including variance, average distance, and Local Intrinsic Dimensionality (LID)~\cite{johnsson2014low,kolacz2016measures}.
% , and a dispersion measure for multi-dimensional data~\cite{kolacz2016measures}. 
However, experimental attempts did not bring about query performance improvements, and even resulted in significant performance loss in some configurations.
We attribute the performance degradation mainly to two factors. 
First, data dispersion measures have difficulty in capturing underlying feature correlations, and intra-subspace measurements discard inter-subspace correlation information.
Second, weighted subspaces convert SC-score from an integer to a real-valued metric, losing the original performance advantage when counting collisions and selecting candidates.
% Therefore, we preliminarily conclude that weighted subspace approaches demonstrate limited efficacy for performance enhancement. 
% Nevertheless, we encourage researchers to explore further, such as using learning-based models to obtain subspace weights. 
%\tp{do we need the last sentence? I would remove}
%\jq{already remove}

In this paper, we adopt the second technical route: performing subspace-oriented data transformation before subspace partitioning.
This route enables comprehensive feature integration across all dimensions during the data transformation, while preserving integer-based SC-score computation, thereby achieving an effective balance between accuracy and efficiency.
The following section presents our subspace-oriented data transformation methodology for realizing data-adaptive subspace collision.

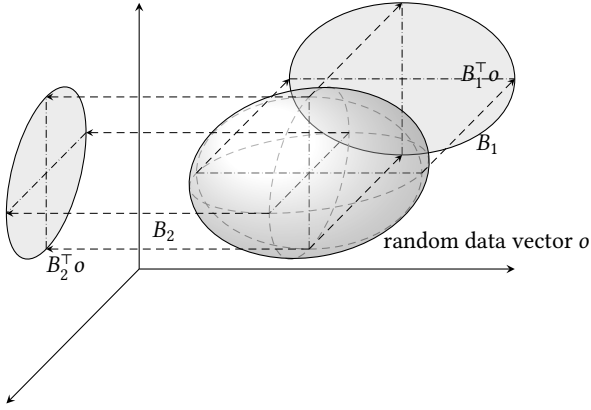
\begin{figure}[tb]
\centering
\small
\begin{tikzpicture}[>=stealth,z={(0:10mm)},x={(-135:5mm)},scale=0.5]
  \draw[->](-7,-5,-7)--++(10,0,0)node[right]{};%{$x$};
  \draw[->](-7,-5,-7)--++(0,0,10)node[above]{};%{$z$};
  \draw[->](-7,-5,-7)--++(0,7,0)node [above]{};%{$y$};
  \begin{scope}[canvas is yz plane at x=-7]
    \draw[fill=lightgray!30](0,0)ellipse(2 and 3);
    \draw[very thin,densely dashdotted](0,-3)--(0,3);
    \draw[very thin,densely dashdotted](-2,0)--(2,0);
  \end{scope}
  \begin{scope}[canvas is xy plane at z=-7]
    \draw[fill=lightgray!30](0,0)ellipse(3 and 2);
    \draw[very thin,densely dashdotted](0,-2)--(0,2);
    \draw[very thin,densely dashdotted](-3,0)--(3,0);
  \end{scope}
  \begin{scope}[canvas is xz plane at y=0]
    % \draw[densely dashed,red](0,0)circle(3);
    \draw[densely dashed,gray](0,0)circle(3);
  \end{scope}
  \begin{scope}[canvas is xy plane at z=0]
    % \draw[densely dashed,blue](0,0)ellipse(3 and 2);
    \draw[densely dashed,gray](0,0)ellipse(3 and 2);
    \draw[very thin,densely dashdotted](0,-2)--(0,2);
    \draw[very thin,densely dashdotted](-3,0)--(3,0);
  \end{scope}
  \begin{scope}[canvas is yz plane at x=0]
    % \draw[densely dashed,green!80!black](0,0)ellipse(2 and 3);
    \draw[densely dashed,gray](0,0)ellipse(2 and 3);
    \draw[very thin,densely dashdotted](0,-3)--(0,3);
    \draw[ball color=white,fill opacity=0.4,rotate=-10](0,0)ellipse(2.21 and 3.22);
  \end{scope}
  \draw[densely dashed,very thin,->](0,0,3)--(-7,0,3)node[midway,below right]{$B_1$};
  \draw[densely dashed,very thin,->](0,0,-3)--(-7,0,-3);
  \draw[densely dashed,very thin,->](0,2,0)--(-7,2,0);
  \draw[densely dashed,very thin,->](0,-2,0)--(-7,-2,0);
  \draw[densely dashed,very thin,->](0,2,0)--(0,2,-7);
  \draw[densely dashed,very thin,->](0,-2,0)--(0,-2,-7);
  \draw[densely dashed,very thin,->](3,0,0)--(3,0,-7)node[pos=0.4,below]{$B_2$};
  \draw[densely dashed,very thin,->](-3,0,0)--(-3,0,-7);
  \draw (5,0,6.5)node{random data vector $o$};
  \draw (7,0,-4)node{$B_2^\top o $};
  \draw (-3,1.5,3.5)node{$B_1^\top o$};
\end{tikzpicture}
\caption{{ Illustration of the proposed subspace-oriented entropy averaging approach, with random data vector $o$ of mean zero and covariance $\Sigma$ represented by a 3D ellipsoid. }}
\label{fig:projection}
\end{figure}

\subsection{Subspace-oriented Data Transformation via Entropy Averaging}
\label{entropy_averaging}

In this section, we describe the motivation and a few theoretical properties of the proposed Subspace-oriented Data Transformation method in Algorithms~\ref{data_transformation}~and~\ref{eigensystem_allocation}.

As motivated above in Section~\ref{subsec:problem_analysis}, here we aim to design a \emph{data-adaptive} subspace collision framework that takes account of the \emph{heterogeneity} in the (different subspaces of the) data.
%To model data having different variability in its coordinates, the simplest and most natural choice is the multivariate Gaussian distribution with non-identity covariance. 
%Consider that the input data follows a multivariate Gaussian distribution of zero mean and covariance $\Sigma$, i.e., $o \sim \mathcal{N}(0,\Sigma)$.
To model data having different variability in its coordinates, we consider that the input data $o$ are of mean zero and covariance $\Sigma$ (that is assumed positive definite and can be \emph{different} from the identity matrix).
Note in particular that we do \emph{not} need to assume the data distribution (e.g., that they are Gaussian), but that the covariance exists.
We aim to find a linear transformation $B \in \mathbb{R}^{d \times (N_s \cdot s)}$ that ``projects'' the input $o$ into $N_s$ subspaces of dimension $s$ each, in a \emph{data-adaptive} way described above.

In this paper, we adopt \emph{differential entropy}~\cite{Cover2005Elements.ch8} from information theory to measure the distributional differences of data in different subspaces.
It is commonly used to quantify the uncertainty or information content of (the distribution of) random data.
Precisely, for input data $o \in \mathbb{R}^{d}$, consider partitioning the linear transformation $B$ into $N_s$ blocks as follow: 
\begin{equation}\label{eq:def_R}
    B = \begin{bmatrix}
    B_1 & B_2 & \ldots & B_{N_s}
  \end{bmatrix},~B_j \in \mathbb{R}^{d \times s},~j \in \{ 1, \ldots, N_s \}.
\end{equation}
Then, the $j^{th}$ subspace after transformation is $B_j^\top o \in \mathbb{R}^{s}$.
If the obtained $N_s$ subspaces after transformation have approximately equal differential entropy, it signifies that the data information is equally spread over all subspaces (quantitatively measured by the area/volume of the ``information ellipsoids'' on the projected subspaces, as shown in Figure~\ref{fig:projection}).

Since the data vector $o$ is of mean zero and covariance $\Sigma$, each its subspace projection $B_j^\top o$ is of mean zero and covariance $B_j^\top \Sigma R_j$.
Our goal is to find a linear transformation $B$ such that, after transformation, the obtained subspaces $B_j^\top o$ are ``uniform'' in the sense of having balanced differential entropies.
This leads to the following constrained min-max optimization problem
\begin{equation}\label{eq:opt_0}
  \begin{aligned}
  \min_j \max_{B_j} & \quad h(B_j^\top o),\\ 
  \text{s.t.}& \quad B_j^\top B_j = I_s,
  \end{aligned}
\end{equation}%
where $h(B_j^\top o)$ is the differential entropy of $B_j^\top o$ that is continuously distributed.
It is known that for any random vector with given mean and covariance, the multi-variate Gaussian distribution maximizes the differential entropy~\cite{Cover2005Elements.ch8}, with
\begin{equation}
  h(B_j^\top o) \propto \log\det(B_j^\top \Sigma B_j),\quad \text{for}\quad B_j^\top o \sim \mathcal{N}(0,\Sigma),
\end{equation}
that is proportional to the log-determinant of the covariance $B_j^\top \Sigma B_j$, provided that $B_j$ has full rank.

As a consequence, the optimization problem in Equation~\eqref{eq:opt_0} can be relaxed into the following min-max problem: 
\begin{equation}\label{eq:opt_1}
  \begin{aligned}
  \min_j \max_{B_j} & \quad \log\det(B_j^\top \hat \Sigma B_j),\\ 
  \text{s.t.}& \quad B_j^\top B_j  = I_s,
  \end{aligned}
\end{equation}%
where we use the sample covariance $\hat \Sigma$ instead of its population counterpart $\Sigma$, and impose the orthonormal constraint so that each block transformation preserves scale: the data covariance is neither artificially amplified nor diminished within each subspace.

\begin{algorithm}[tb]
% \small
% \footnotesize
	\caption{Subspace-oriented Data Transformation}                                
	\label{data_transformation}
	\LinesNumbered
	\KwIn{Dataset $\mathcal D = \{o_1,o_2,\ldots,o_n\}$, dataset size $n$, data dimensionality $d$, subspace number $N_s$, subspace dimensionality $s$}
	\KwOut{Transformed dataset $\mathcal T$}
        Initialize the transformed dataset $\mathcal T$ with size $(n, N_s \cdot s)$; \\
        Compute the mean value of $\mathcal D$: $\bar o = \frac{1}{n}\sum^{n}_{i=1}o_i$; \\
	Compute the covariance matrix of all points: $\hat \Sigma=\frac{1}{n-1}\sum^{n}_{i=1}(o_i-\bar o)(o_i-\bar o)^\top$; \\
        Perform the spectral decomposition of $\hat \Sigma$, obtain the eigenvectors $\mathcal E = \{\xi_1, \xi_2, \ldots, \xi_d\}$ and their corresponding eigenvalues $\Lambda = \{\lambda_1, \lambda_2, \ldots, \lambda_d\}$; \\
        $B \leftarrow$ \textbf{call} \emph{Eigensystem Allocation$(\mathcal E, \Lambda, d, N_s, s)$}; \\
               
	\For{$i=1$ to $n$}{
        \For{$j=1$ to $N_s$}{ 
        Obtain the eigenvectors assigned to the $j$-th subspace: $B_j=\{\xi_1^j, \xi_2^j, \ldots, \xi_s^j\}$; \\
            For any $o_i \in \mathbb{R}^d$, its transformed representation at the $j$-th subspace is: $r_i^j \leftarrow B_j^\top (o_i-\bar o) = ((\xi_1^{j})^\top(o_i-\bar o), (\xi_2^{j})^\top(o_i-\bar o), \ldots, (\xi_s^{j})^\top(o_i-\bar o)) \in \mathbb{R}^s$; \\
        }
        Concatenate the transformed representations of $o_i$ in $N_s$ subspaces $(r_i^1,r_i^2,\ldots,r_i^{N_s})$ to obtain the transformed data point $o_i^\prime \in \mathbb{R}^{N_s \cdot s}$; \\
        $\mathcal T_i \leftarrow o_i^\prime$; \\
	}
	\Return the transformed dataset $\mathcal T$; \\
\end{algorithm}

The min-max optimization problem in \eqref{eq:opt_1} can be effectively solved using the Eigensystem Allocation method described in Algorithm~\ref{eigensystem_allocation}, per the following result.
% \textcolor{red}{
% This is also empirically supported by Gaussianity tests of the five datasets used in the paper, where we observe that GIST1M and SIFT10M data are far from Gaussian, while DEEP1M, Yandex Deep10M, and Microsoft SPACEV10M are closer to Gaussian.
% As shown in Section~\ref{experimental_evaluation}, TaCo demonstrates strong performance on all five Gaussian and non-Gaussian datasets.
% }

\begin{theorem}[Performance Guarantee for Algorithm~\ref{eigensystem_allocation}]\label{theo:main}
Assume that the data sample covariance $\hat \Sigma$ has distinct eigenvalues that are all greater than or equal to one.
Then, the Eigensystem Allocation method in Algorithm~\ref{eigensystem_allocation} solves the optimization problem in Equation~\eqref{eq:opt_1}.
\end{theorem}
\begin{proof}[Proof of Theorem~\ref{theo:main}]
To show that Eigensystem Allocation in Algorithm~\ref{eigensystem_allocation} solves the optimization problem in \eqref{eq:opt_1}, we first focus on the inner maximization problem:
Denote $\lambda_1 > \ldots > \lambda_d$ and $\mu_1 \geq \ldots \geq \mu_s$ the eigenvalues of $\hat \Sigma$ and $B_j^\top \hat \Sigma B_j$, respectively (in descending order), we have, by the Poincaré separation theorem (also known as general Cauchy interlacing theorem) that 
\begin{equation}\label{eq:interlacing}
  \lambda_i \geq \mu_i \geq \lambda_{d - s + i}, \quad i = 1, \ldots, s,
\end{equation}
that is, each eigenvalue $\mu_i$ of $B_j^\top \hat \Sigma B_j$ lies between two eigenvalues of $\hat \Sigma$.
Since the logarithm function is monotonically increasing, maximizing the (log-)determinant in \eqref{eq:opt_1} is equivalent to maximizing the product of eigenvalues.
We formally have, by \eqref{eq:interlacing}, that
\begin{equation}
   \max_{B_j^\top B_j = I_s} \det(B_j^\top \hat \Sigma B_j) = \prod_{i=1}^s \lambda_i,
 \end{equation} 
and the maximum is achieved by taking $B_j$ to be the associated eigenvectors, i.e., $B_j = [\xi_1, \ldots, \xi_s]$.

Further note that the greedy procedure in Algorithm~\ref{eigensystem_allocation} finds the optimal ``balanced'' partition $B_1, \ldots, B_{N_s}$ that solves the outer minimization problem of \eqref{eq:opt_1}.
This concludes the proof. 
% of Theorem~\ref{theo:main}.
\end{proof}

\begin{algorithm}[tb]
% \small
% \footnotesize
	\caption{Eigensystem Allocation}                 
	\label{eigensystem_allocation}
	\LinesNumbered
	\KwIn{Eigenvectors $\mathcal E = \{\xi_1, \xi_2, \ldots, \xi_d\}$ and eigenvalues $\Lambda = \{\lambda_1, \lambda_2, \ldots, \lambda_d\}$ of $\hat \Sigma$, data dimensionality $d$, subspace number $N_s$, subspace dimensionality $s$}
	\KwOut{A set of buckets holding eigenvectors assigned to each subspace $B = \{B_1, B_2, \ldots, B_{N_s}\}$ }
    Initialize a set of $N_s$ empty buckets $B = \{B_1, B_2, \ldots, B_{N_s}\}$ to hold the eigenvectors assigned to each subspace; \\
    Initialize $N_s$ values $p_1 = p_2 = \ldots = p_{N_s} = 1$ to hold the product of the eigenvalues assigned to each subspace; \\
    Scale the eigenvalues proportionally such that $\forall \lambda_i \in \Lambda, \; \lambda_i \geq 1$; \\
    Reorder the eigenvalues and eigenvectors in descending order: $\mathcal E^\prime = \{\xi_1^\prime, \xi_2^\prime, \ldots, \xi_d^\prime\}$, $\Lambda^\prime = \{\lambda_1^\prime, \lambda_2^\prime, \ldots, \lambda_d^\prime\}$, such that $\lambda_1^\prime \geq \lambda_2^\prime \geq \ldots \geq \lambda_d^\prime \geq 1$; \\
        
	\For{$i=1$ to $N_s \cdot s$}{
        $bucket \leftarrow \mathop{\mathrm{argmin}}\limits_{j \in \{1,2,\ldots,N_s\},\; |B_j| < s} p_j$; \\
        $B_{bucket} \leftarrow B_{bucket} \cup \{\xi_i^\prime\}$; \\
        $p_{bucket} \leftarrow p_{bucket} \cdot \lambda_i^\prime$; \\
	}
	\Return the eigenvectors allocation set $B$; \\
\end{algorithm}

\subsection{SC-score Distribution after Transformation}

In this section, we explore the neighborhood relationship of the data after the subspace-oriented transformation.
We first analyze the SC-score distributions of the transformed data on four datasets under a unified experimental setting, with $N_s = 6$, $s = 8$, and $\alpha = 0.05$.
Figure~\ref{scscore_transformed} shows that the SC-score of the transformed data still follows the \emph{Pareto principle}, exhibiting the same statistical behavior as the original (untransformed) data shown in Figure~\ref{scscore}.
Experimental results indicate that, the transformed data preserves the good neighborhood relationships
of the original data.

\begin{figure} [tb]
% \vspace*{-0.3cm}
	% \flushleft
        \subfigcapskip=5pt
        \subfigure[DEEP1M]{
		\includegraphics[width=0.32\linewidth]{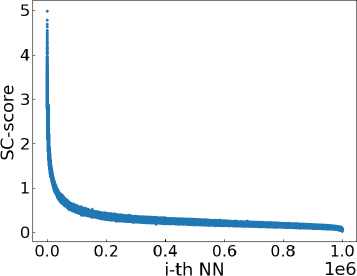}
		\label{deep1m_transformed}}\hspace{2mm}
	\subfigure[GIST1M]{
		\includegraphics[width=0.32\linewidth]{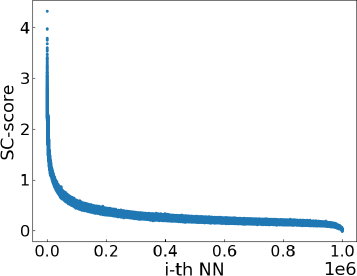}
		\label{gist1m_transformed}}\hspace{2mm}
	\subfigure[Yandex DEEP10M]{
		\includegraphics[width=0.32\linewidth]{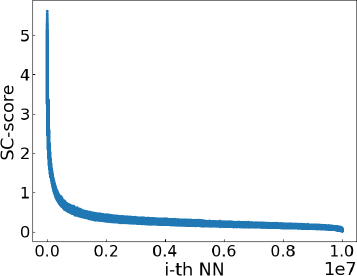}
		\label{deep10m_transformed}}\hspace{2mm}
        \subfigure[Microsoft SPACEV10M]{
		\includegraphics[width=0.32\linewidth]{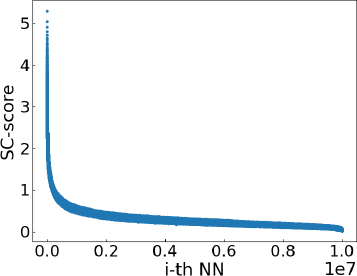}
		\label{spacev10m_transformed}}
  % \vspace*{-0.3cm}
	\caption{SC-score of the transformed data still follows the Pareto principle as the original data in Figure~\ref{scscore}.}
	\label{scscore_transformed}
 % \vspace*{-0.3cm}
\end{figure}

To investigate the underlying mechanism behind this phenomenon, in the remainder of this section, we perform rigorous theoretical analyses showing that the transformation introduced in Algorithm~\ref{eigensystem_allocation} preserves (pairwise) local distances (Lemma~\ref{lemma:distance}), and consequently the relative neighborhood ordering (Theorem~\ref{theo:relative_order}).
% Lemma~\ref{lemma:distance} shows that if local difference vectors are well captured by the allocated eigenvectors from Algorithm~\ref{eigensystem_allocation}, then Algorithm~\ref{data_transformation} induces only bounded distortion on pairwise local distances.
\begin{lemma}[Local Distance Preservation for Algorithm~\ref{eigensystem_allocation}]\label{lemma:distance}
Let $\mathcal D=\{o_i\}_{i=1}^n \in \mathbb{R}^d$ be a dataset, and let
$B \in \mathbb{R}^{ d \times (N_s \cdot s)}$ denote the transformation returned by Algorithm~\ref{eigensystem_allocation}, i.e., $B=
  \begin{bmatrix}
    B_1 & B_2 & \ldots & B_{N_s}
  \end{bmatrix}
$, with $B_j \in \mathbb{R}^{d \times s}$ and $N_s \cdot s \leq d$.
For any point $o_i \in \mathcal{D}$ and its neighbor
$o_j$, assume that the local difference vector
satisfies
\begin{equation}
\label{eq:deterministic_assumption}
\|(I_d - B B^\top)(o_i - o_j)\|^2
\le
\varepsilon \, \|o_i - o_j\|^2,
\end{equation}
for some $\varepsilon \in (0,1)$.
Then, the distance after transformation satisfies\footnote{The assumption of the form in \eqref{eq:deterministic_assumption} has been established for random data points drawn from the so-called spiked random matrix model, see~\cite{paul2007asymptotics,johnstone2018pca,couillet2022RMT4ML}. }
\begin{equation}
\label{eq:deterministic_bound}
(1-\varepsilon)\,\|o_i - o_j\|^2
\le
\|B^\top (o_i - o_j)\|^2
\le
\|o_i - o_j\|^2.
\end{equation}
\end{lemma}
\begin{proof}[Proof of Lemma~\ref{lemma:distance}]
Recall from Algorithm~\ref{eigensystem_allocation} that $B \in \mathbb{R}^{ d \times (N_s \cdot s)}$ contains eigenvectors of $\hat \Sigma$ as its columns, so that $B B^\top$ is an orthogonal projection matrix. 
As a consequence, $o_i - o_j \in \mathbb{R}^d$ can be
decomposed as the sum of the projection $B B^\top(o_i - o_j)$ and the corresponding residue $(I_d - B B^\top)(o_i - o_j)$, with
% into two orthogonal components lying in the retained and discarded subspaces, respectively. 
% Therefore,
\begin{equation}
\label{eq:orth_decomposition}
\|o_i - o_j\|^2
=
\|B^\top(o_i - o_j)\|^2
+
\|(I_d - B B^\top)(o_i - o_j)\|^2.
\end{equation}
Rearranging~\eqref{eq:orth_decomposition} gives
\begin{equation}
\label{eq:rearranged}
\|B^\top (o_i - o_j)\|^2
=
\|o_i - o_j\|^2
-
\|(I_d - B B^\top)(o_i - o_j)\|^2.
\end{equation}
By~~\eqref{eq:deterministic_assumption}, the second term
on the right-hand side is bounded by
$\varepsilon \|o_i - o_j\|^2$, which yields the lower bound in
\eqref{eq:deterministic_bound}.
The upper bound follows directly from the non-expansiveness of
orthogonal projection.
% i.e., $\|P(o_i - o_j)\| \le \|o_i - o_j\|$.
% This concludes the proof of Lemma~\ref{lemma:distance}.
\end{proof}
With Lemma~\ref{lemma:distance} at hand, we show next in Theorem~\ref{theo:relative_order} that Algorithm~\ref{data_transformation} is guaranteed to preserve relative neighborhood ordering under the conditions of Lemma~\ref{lemma:distance}.
\begin{theorem}[Relative Neighborhood Ordering Preservation for Algorithm~\ref{data_transformation}]
\label{theo:relative_order}
Under the notations and conditions of Lemma~\ref{lemma:distance}, consider two points
$o_j, o_z \in \mathcal{D}$ such that
\begin{equation}
\label{eq:relative_condition}
\|o_i - o_j\|^2
<
(1-\varepsilon)\,\|o_i - o_z\|^2.
\end{equation}
Then, their relative ordering with respect to $o_i$ is preserved after
the transformation through $B$, in the sense that
\begin{equation}
\|B^\top(o_i - o_j)\|
<
\|B^\top(o_i - o_z)\|.
\end{equation}
\end{theorem}
\begin{proof} [Proof of Theorem~\ref{theo:relative_order}]
It follows from Lemma~\ref{lemma:distance} that
\[
\|B^\top(o_i - o_j)\|^2
\le
\|o_i - o_j\|^2
\quad\text{and}\quad
\|B^\top(o_i - o_z)\|^2
\ge
(1-\varepsilon)\,\|o_i - o_z\|^2.
\]
Combining these two inequalities with
\eqref{eq:relative_condition} yields
\[
\|B^\top(o_i - o_j)\|^2
<
\|B^\top(o_i - o_z)\|^2.
\]
This concludes the proof of Theorem~\ref{theo:relative_order}.
\end{proof}

\begin{figure}[tb] 
	\centering
	\includegraphics[width=0.9\linewidth]{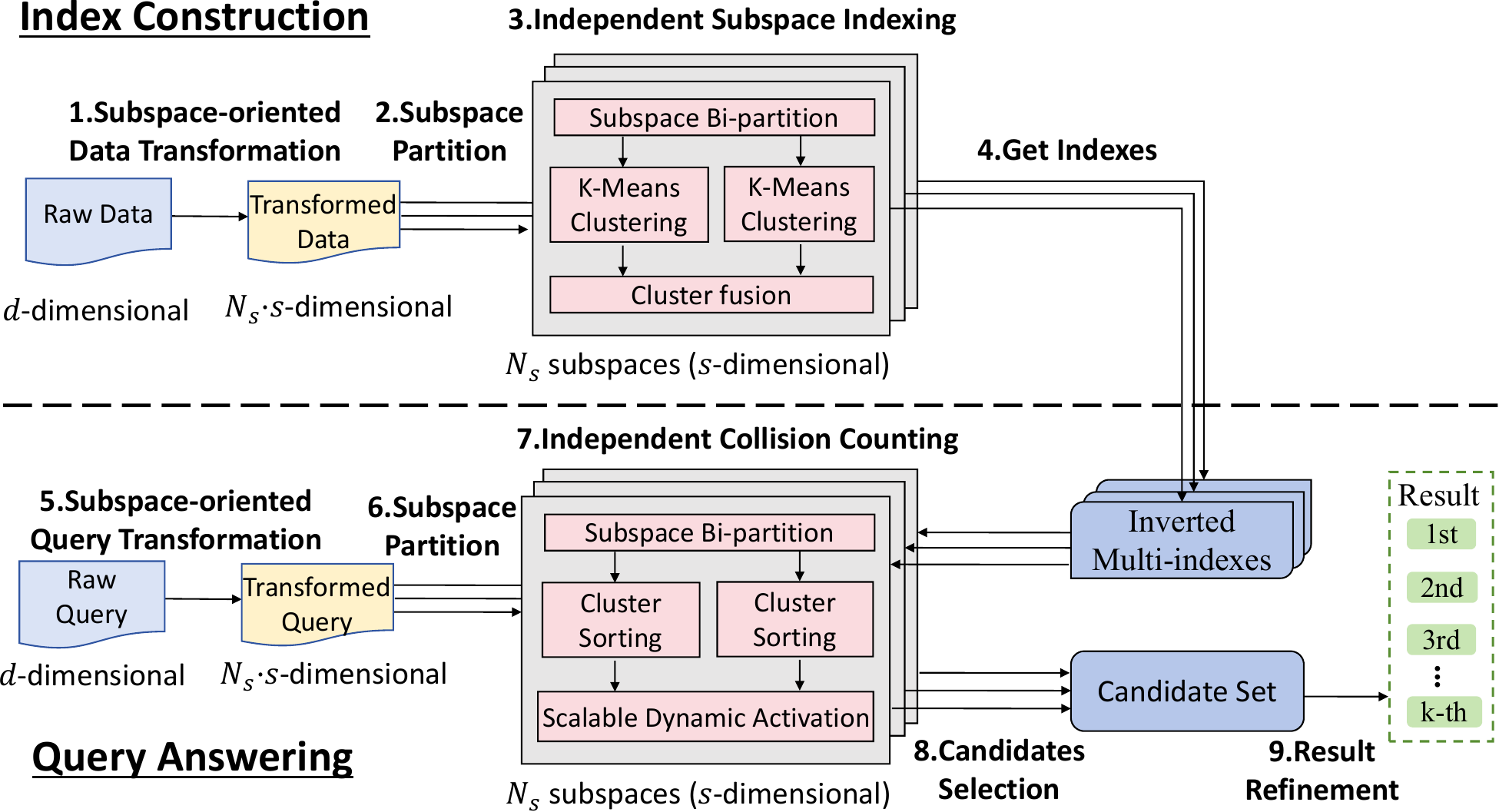}
  % \vspace*{-0.3cm}
	\caption{Overview of the TaCo workflow.}
	\label{overview}
  % \vspace*{-0.4cm}
\end{figure}

\section{The TaCo Method} \label{our_method}

% Section~\ref{data_adaptive_mechanism} presents subspace-oriented data transformation, enabling the subspace collision framework to become data-adaptive.
% A further requirement is to design a $k$-ANNS method for the new framework to achieve efficient index construction and query answering.

In this section, we describe a data-adaptive and query-aware subspace collision method called TaCo.
Figure~\ref{overview} illustrates the TaCo workflow.
During index construction, subspace-oriented data transformation is first applied, reducing data dimensionality from $d$ to $N_s \cdot s$.
The transformed data is then partitioned into $N_s$ subspaces, and an inverted multi-index (IMI) is constructed in each subspace.
During query answering, the query undergoes transformation consistent with the index construction and is subsequently partitioned into $N_s$ subspaces.
The novel \emph{Scalable Dynamic Activation} algorithm is then executed on each subspace's IMI to enable efficient collision counting.
Next, the \emph{Query-aware Candidates Selection} algorithm identifies candidate points based on the collision distribution of each query.
The final step involves refining the candidate set and returning the top-$k$ results.

\subsection{Index Construction} \label{indexing_phase}

As analyzed in Section~\ref{data_adaptive_mechanism}, subspace-oriented data transformation is a prerequisite for adapting the subspace collision framework to the data distribution.
This process is formalized in Algorithm~\ref{data_transformation}.
First, compute the mean vector and covariance matrix of data points, and perform the spectral decomposition to obtain eigenvalues and eigenvectors (lines~2-4).
Algorithm~\ref{eigensystem_allocation} (detailed subsequently) is then invoked to generate the subspace-oriented eigenvector allocation set (line~5).
Within each subspace, all data points are transformed using the assigned $s$ eigenvectors (lines~6-9).
The concatenation of these subspace-specific transformed representations produces the final transformed dataset of dimensionality $N_s \cdot s$ (lines~10-12).

Algorithm~\ref{eigensystem_allocation} presents the eigensystem allocation procedure, which provides the optimal solution for balanced subspace partition, as proven in Section~\ref{entropy_averaging}.
Specifically, $N_s$ eigenvector buckets and eigenvalue product trackers are created (lines~1-2). 
Eigenvalues are scaled for ease of calculation (line~3), then reordered with corresponding eigenvectors in descending sequence (line~4).
The eigensystem allocation sequentially processes the top $N_s \cdot s$ eigenvectors in descending eigenvalue order (line~5). 
For each eigenvector being processed (corresponding to the largest available eigenvalue), the bucket with the minimum product of the eigenvalues is selected. 
Each bucket contains no more than $s$ eigenvectors (line~6).
The eigenvector is then assigned to the selected bucket and the corresponding eigenvalue product tracker is updated (lines~7-8).
The finalized eigenvector allocation set is returned (line~9).

\begin{algorithm}[tb]
% \small
% \footnotesize
	\caption{Create Index}                                                                           
	\label{create_index}
	\LinesNumbered
	\KwIn{Dataset $\mathcal D$, dataset size $n$, data dimensionality $d$, subspace number $N_s$, subspace dimensionality $s$, number of K-means clusters $K$, K-means iterations $t$}
	\KwOut{Centroid list $centroids$ and inverted multi-index list $IMIs$}
    $\mathcal T \leftarrow$ \textbf{call} \emph{Subspace-oriented Data Transformation$(\mathcal D, n, d, N_s, s)$}; \\
	Divide the transformed dataset $\mathcal T$ into $N_s$ subspaces: $\mathcal T_1, \mathcal T_2,\ldots, \mathcal T_{N_s}$, each with $s$ dimensions; \\
    Initialize a list $centroids$ of length $2N_s$ and a list $IMIs$ of length $N_s$; \\
	\For{$i=1$ to $N_s$}{
            Initialize a map $IMI_i$ as the inverted multi-index; \\
            Further split $\mathcal T_i$ into two parts: $\mathcal T_i^1$ and $\mathcal T_i^2$; \\
            $centroids^1_i, assignments^1_i \leftarrow$ \textbf{call} \emph{Kmeans$(\mathcal T_i^1, \sqrt{K}, t)$}; \\
            $centroids^2_i, assignments^2_i \leftarrow$ \textbf{call} \emph{Kmeans$(\mathcal T^2_i, \sqrt{K}, t)$}; \\
            $centroids.append(centroids^1_i, centroids^2_i)$; \\
            \For{$j=0$ to $n-1$}{
            $IMI_i[assignments^1_i[j], assignments^2_i[j]].append(j)$; \\
            }
            $IMIs.append(IMI_i)$; \\
	}
	\Return $centroids$ and $IMIs$; \\
\end{algorithm}

Figure~\ref{overview} illustrates the workflow of TaCo for index construction, while Algorithm~\ref{create_index} provides the corresponding pseudocode.
The process begins by invoking the \emph{Subspace-oriented Data Transformation} algorithm to the original dataset, resulting in a transformed dataset with reduced dimensionality $N_s \cdot s$ (line~1).
This transformed dataset is then partitioned into $N_s$ subspaces, each of dimensionality $s$ (line~2).
% TaCo's index comprises two components: centroids and inverted multi-indexes.
% Accordingly, the initialization of two lists is required to hold these elements during index construction (line~3).
For each subspace, the algorithm constructs a separate inverted multi-index (lines~4-12). 
First, the subspace data is split into two disjoint parts in terms of dimension (line~6). 
Each part is then clustered using K-means with $\sqrt{K}$ centroids and $t$ iterations, generating two sets of centroids and corresponding assignment labels for all data points (lines~7-8). 
% The centroids from both parts are stored in the centroid list (line~9).
To enable efficient retrieval during query processing, a map structure is constructed as the inverted multi-index for each subspace.
For each point, its assignment labels from both clustering operations form a composite key within the map, under which the point's original identifier is appended to the corresponding cell (lines~10-12). 
% Finally, the computed centroids and the constructed IMI lists are returned.

We note that the main difference between TaCo and SuCo~\cite{wei2025subspace} in index construction is that TaCo achieves balanced subspace partitioning by invoking the \emph{Subspace-oriented Data Transformation} function (Algorithm~\ref{data_transformation}), which addresses the data-agnostic subspace partitioning problem encountered by SuCo's index structure.
The subsequent steps of building an IMI within each subspace follow the same procedure in both TaCo and SuCo.

\begin{algorithm}[tb]
% \small 
% \footnotesize
	\caption{Scalable Dynamic Activation}
	\label{scalable_dynamic_activation}
	\LinesNumbered
	\KwIn{Collision ratio $\alpha$, dataset size $n$, number of K-means clusters $K$, distances and indices of the first and second parts in a subspace $dists_1,idx_1,dists_2,idx_2$, the inverted multi-index $IMI$}
	\KwOut{Clusters containing data points that collide with $q$}
        Initialize a list $retrieved\_clusters$, a min heap $active\_clusters$, an array $active\_idx$ of length $\sqrt{K}$ and set to 0; \\
        $retrieved\_num \leftarrow 0$; \\
        $active\_clusters.push(pair<dists_1[idx_1[0]]+dists_2[idx_2[0]], 0>)$; \\
        \While{\textit{TRUE}}{
            $selected\_cell \leftarrow active\_clusters.top()$; \\
            $pos \leftarrow selected\_cell.second$; \\
            $cluster \leftarrow IMI[idx_1[pos], idx_2[active\_idx[pos]]]$; \\
            $retrieved\_clusters.append(cluster)$; \\
            $retrieved\_num$ += $sizeof(cluster)$; \\
            \If{$retrieved\_num \geq \alpha \cdot n$}{
                break; \\
            }
            \If{$active\_idx[pos]==0$ and $pos<\sqrt{K}-1$}{
                $active\_clusters.push(pair<dists_1[idx_1[pos+1]]+dists_2[idx_2[0]], pos+1>)$; \\
            }
            $active\_clusters.pop()$; \\
            \If{$active\_idx[pos]<\sqrt{K}-1$}{
                $active\_idx[pos]$++; \\
                $selected\_cell.first \leftarrow dists_1[idx_1[pos]]+dists_2[idx_2[active\_idx[pos]]]$; \\
                $active\_clusters.push(selected\_cell)$; \\
            }
        }
        
	\Return $retrieved\_clusters$; \\
\end{algorithm}

% \subsection{Query Answering} \label{query_phase}

% As illustrated in Figure~\ref{overview_query}, the query processing pipeline~consists of three core stages:
% (1) collision counting via the inverted multi-index in each subspace (Algorithm~\ref{scalable_dynamic_activation});
% (2) candidate selection based on the aggregated collision distribution across all subspaces (Algorithm~\ref{candidate_selection}); and
% (3) refinement and extraction of the final $k$-ANNS results (Algorithm~\ref{kann_query}).

\subsection{Collision Counting}
\label{collision_counting}

The \emph{Dynamic Activation} algorithm was proposed along with SuCo to facilitate efficient collision counting~\cite{wei2025subspace}. 
However, its scalability remains limited.
The algorithm maintains an activation list via a linear array, resulting in $\mathcal{O}(l)$ query complexity and $\mathcal{O}(1)$ update complexity, where $l = \sqrt{K}$ denotes the length of the IMI list.
While this design offers satisfactory performance when $l$ is small, it becomes a bottleneck when a larger $l$ is required to improve indexing precision.
The linear query complexity constrains overall efficiency under such conditions.
To address this limitation, we propose the \emph{Scalable Dynamic Activation} algorithm (Algorithm~\ref{scalable_dynamic_activation}).
Experimental results in Section~\ref{imiquery_compare} demonstrate that the \emph{Scalable Dynamic Activation} algorithm improves efficiency by up to 30\% compared to the original \emph{Dynamic Activation} algorithm.
The \emph{Scalable Dynamic Activation} algorithm is not only applicable to TaCo, but can also be migrated to \emph{any} method that uses the IMI index structure.

Algorithm~\ref{scalable_dynamic_activation} gives the pseudocode of the \emph{Scalable Dynamic Activation} algorithm.
% which efficiently retrieves clusters from the IMI that collide with a query~$q$, based on a predefined collision ratio~$\alpha$.
% Given the distances between the query and all centroids in the two subspaces, $dists_1,dists_2$, along with the corresponding index arrays $idx^{i}_1,idx^{i}_2$ that record the sorted order of these distances, the algorithm systematically combines clusters from both subspaces. 
% Each combined cluster in IMI is evaluated using the sum of the distances from the two subspaces, with smaller sums having higher priority for retrieval.
It begins by pushing the initial cluster pair with the smallest distance sum into a min heap (line~3). 
In each iteration, the cluster combination with the smallest current distance is popped from the min heap and added to the result set (lines~6-9). 
The number of data points retrieved is accumulated, and the algorithm terminates once the accumulated count meets the threshold (lines~10-11).
If the current cluster from the first subspace is being activated for the first time and subsequent clusters remain available, the next cluster is activated and pushed into the min heap (lines~12-13). 
The current cluster’s activation index is incremented, and its next combination in the second subspace is computed and pushed back into the min heap for further evaluation (lines~14-18). 
The algorithm returns all retrieved clusters once the collision requirement is satisfied (line~19).
This mechanism ensures that cluster combinations are retrieved in ascending order of their distance sums, enabling \emph{early termination} when sufficient collisions are collected.

The main difference between TaCo and SuCo in the \emph{Collision Counting} step is that TaCo adopts a min-heap to manage the activation list in its \emph{Scalable Dynamic Activation} algorithm, achieving $\mathcal{O}(1)$ query complexity, thereby addressing the linear query complexity problem faced by SuCo's \emph{Dynamic Activation} algorithm. 
This new data structure introduces corresponding changes to the algorithm logic, including the use of a pair structure (lines 3, 6, 13, and 17 in Algorithm~\ref{scalable_dynamic_activation}) and the management of the heap (lines 3, 5, 13, and 18 in Algorithm~\ref{scalable_dynamic_activation}).
The logic of sequentially retrieving data points by dynamically activating cells in the IMI is the same for both TaCo and SuCo.

\subsection{Candidate Selection}

\subsubsection{Overview}

The query strategies of SuCo are \emph{query-oblivious}, employing same query overhead for all queries~\cite{wei2025subspace}.
To enhance query efficiency, we aim to design \emph{query-aware} query strategies for TaCo.
The entire query pipeline offers two opportunities for per-query customization: collision counting and candidate selection.
Collision counting already leverages an \emph{early-termination} strategy through the \emph{Scalable Dynamic Activation} algorithm, achieving optimal overhead adaptation for distinct queries.
However, designing a \emph{query-aware} candidate selection strategy remains an open problem.
Through statistical analysis of intermediate query states (specifically, the distribution of SC-scores), we observe that certain queries exhibit highly discriminative SC-score distributions, where high recall can be achieved by re-ranking only the top-scoring data points. 
In contrast, other queries show poor SC-score discriminability, requiring a larger candidate set for re-ranking to achieve comparable recall.
Thus, unlike SuCo, which selects a fixed number $\beta \cdot n$ of candidates for all queries, the \emph{Query-aware Candidates Selection} algorithm designed for TaCo dynamically determines the candidate set size per query based on its SC-score distribution.

\begin{algorithm}[tb]
% \small
% \footnotesize
	\caption{Query-aware Candidates Selection}
	\label{candidate_selection}
	\LinesNumbered
	\KwIn{Dataset $\mathcal D = \{o_1,o_2,\ldots,o_n\}$, SC-score distribution: $SC\_scores$, re-rank ratio $\beta$, subspace number $N_s$, dataset size $n$}
	\KwOut{Selected candidates $\mathcal C$ and the candidate number}
        Initialize an empty candidate set $\mathcal C \leftarrow \varnothing$; \\
        Initialize an array $collision\_num$ of size $N_s + 1$ and set to 0; \\
        \For{$i=1$ to $n$}{
            $collision\_num[SC\_scores[i-1]]$++; \\
	   }
       $last\_collision \leftarrow N_s$; \\
       $candidate\_num \leftarrow 0$; \\
       \For{$j=N_s$ to $0$}{
       $candidate\_num$ += $collision\_num[j]$; \\
            \If{$collision\_num[j] \leq \beta \cdot n - candidate\_num$}{
                $last\_collision$ -= 1; \\
            }\Else{
                break; \\
            }
	   }
       \For{$i=1$ to $n$}{
             \If{$SC\_scores[i-1] \geq last\_collision$}{
                $\mathcal C= \mathcal C \cup \{o_i\}$; \\
            }
	   }
	\Return the candidate set $\mathcal C$ and $candidate\_num$; \\ 
\end{algorithm}

\subsubsection{Algorithm Design}

Algorithm~\ref{candidate_selection} provides the pseudocode of the \emph{Query-aware Candidates Selection} algorithm, which identifies candidate points based on the collision distribution of each query.
The algorithm begins by initializing an empty candidate set $\mathcal C$ and an array $collision\_num$ of size $N_s + 1$ to count the number of data points for each possible SC-score (lines~1-2). 
It then iterates through all data points to populate $collision\_num$ according to their SC-scores (lines~3-4).
Next, the algorithm determines a threshold $last\_collision$ starting from the maximum possible score $N_s$ (line~5). 
It selects candidates from higher to lower SC-scores (lines~7-8) until the accumulated number of candidates exceed $\beta \cdot n$ (lines~9-12). 
The algorithm collects all data points whose SC-scores are no less than $last\_collision$, forming the candidate set $\mathcal C$ (lines~13-15). 
Finally, it returns $\mathcal C$ together with the total number of candidates (line~16).

\subsubsection{Theoretical Analysis}

We now provide a theoretical justification for how the query discriminability affects SC-score distribution and leads to different candidate selection behaviors.

\begin{definition}[Query Discriminability]%[Query Discriminability]
\label{def:quer_discriminability}
Given a query $q$ and the number of subspaces $N_s$.
As a consequence of the uniformity of the subspaces after the subspace-oriented data transformation in Theorem~\ref{theo:main}, we denote $p^*$ the \emph{collision probability} such that subspace collision happens for a true nearest neighbor $o^* \in \mathcal{R}^*= \{o^*_1,\ldots,o^*_k\}$, and $p$ the probability that subspace collision happens for a non-neighbor $o \notin \mathcal{R}^*$, with $p^* > p$.
The Discriminability Gap of $q$ is defined as $\Delta(q) = p^* - p > 0$.
% Benefiting from the subspace-oriented data transformation (Theorem~\ref{theo:main}), each subspace is uniform and independent.
% Let $p_i$ be the collision probability in the $i$-th subspace. We define $p = \frac{1}{N_s} \sum p_i$ as the expected average collision probability, which characterizes the aggregate behavior of the query.
% Let $p^*$ and $p$ be the average collision probabilities for a nearest neighbor $o^* \in \mathcal{R}^*= \{o^*_1,\ldots,o^*_k\}$ and a non-neighbor $o \notin \mathcal{R}^*$ over $N_s$ subspaces, respectively. The Discriminability Gap is defined as $\Delta(q) = p^* - p$.
% Given a maximum tolerance for classification error $\delta \in (0, 1)$, we derive a Minimum Separation Threshold $\tau_\delta = \sqrt{\frac{2 \ln(1/\delta)}{N_s}}$ based on concentration bounds. 
% We say that $q$ is \emph{strongly discriminative} if $\Delta(q) \ge \tau_\delta$, and $q$ is \emph{weakly discriminative} if $\Delta(q) < \tau_\delta$.
\end{definition}
Intuitively, for queries with larger discriminability gap $\Delta(q)$, the SC-score of a true nearest neighbor $SC(o)$ and that of a non-neighbor $SC(o^*)$ are further separated than queries with smaller discriminability gap.
As a consequence, the Type-I error (i.e., the probability that a non-neighbor is incorrectly classified as a neighbor) and Type-II error (i.e., the probability that a true neighbor is incorrectly classified as a non-neighbor) are expected to decrease rapidly as $\Delta(q)$ grows large.
This intuition is made precise in the following result.
\begin{lemma}[SC-score Separation] \label{lemma:separation}
For a given query $q$ and its discriminability gap $\Delta(q)$ as in Definition~\ref{def:quer_discriminability}, denote $o^*$ a true nearest neighbor of $q$ and $o \notin \mathcal{R}^*$ a non-neighbor.
Then, the optimal Type-I and Type-II errors both decay at least at the rate $\exp\left( - \frac{N_s \Delta^2(q)}{8p(1-p)} \right)$.
That is, both errors decay exponentially fast to zero as $N_s$ or $\Delta(q)$ becomes large.
% Then, we have that $SC(o)$ the SC-score of $o$ follows a binomial distribution with parameters $N_s$ and $p$, and $SC(o^*)$ the SC-score of $o^*$ follows a binomial distribution with parameters $N_s$ and $p^*$.
% In this case, the Bayes optimal decision threshold (with equal priors and equal costs) is given by $\frac{-N_s\log\frac{1-p^*}{1-p}} {\log\Big(\frac{p^*(1-p)}{p(1-p^*)}\Big)}$, and the corresponding Type-I and Type-II error given by XX and XX, respectively.
% $\left\lceil \frac{-N_s\log\frac{1-p^*}{1-p}} {\log\Big(\frac{p^*(1-p)}{p(1-p^*)}\Big)} \right\rceil.$
% Let $SC(o)$ be the SC-score of a point $o$. We analyze the probability of classification error $\Pr(\text{Error})$ defined as the overlap between the lower tail of neighbors and the upper tail of non-neighbors at the optimal decision boundary $t = \frac{N_s(p+p^*)}{2}$. The behavior differs based on discriminability:
% (1) If $q$ is strongly discriminative, the overlap probability is bounded exponentially by the square of the gap, ensuring minimal error:$\Pr(\text{Error}) \le \exp\left(-\frac{1}{2} N_s \Delta(q)^2\right) \le \delta$. This implies that high SC-scores are dominated almost exclusively by true neighbors.
% (2) If $q$ is weakly discriminative, the bound becomes loose:$\delta < \exp\left(-\frac{1}{2} N_s \Delta(q)^2\right) \le 1$. Specifically, if $\Delta(q) \to 0$ or $\Delta(q) \le 0$, the exponential term approaches 1, indicating that $SC(o^*)$ and $SC(o)$ are statistically mixed, and a simple threshold $t$ cannot effectively separate them.
\end{lemma}
\begin{proof}[Proof of Lemma~\ref{lemma:separation}]
As a consequence of Theorem~\ref{theo:main}, we have that $SC(o)$ the SC-score of $o$ follows a binomial distribution with parameters $N_s$ and $p$, and $SC(o^*)$ the SC-score of $o^*$ follows a binomial distribution with parameters $N_s$ and $p^*$.
The associated Bernoulli KL divergence writes, for small $\Delta(q)$ as
\begin{equation}
    D(p||p^*) = \frac{\Delta^2(q)}{2p(1-p)} + O(\Delta^3(q)).
\end{equation}
Using the Bayes optimal decision rule (that balances the two log-likelihoods) and applying the (relative entropy) Chernoff bound, we conclude the proof of Lemma~\ref{lemma:separation}.
% we have that $SC(o) \sum_{i=1}^{N_s} X_i$, with $X_i$ independent Bernoulli trials $X_1, \dots, X_{N_s}$. Note that while each $p_i$ may vary slightly, Hoeffding’s Inequality holds for any sum of independent bounded variables, allowing us to bound the deviation from the aggregate mean $N_s p$.
% Let the decision threshold be the midpoint $t = \frac{N_s(p + p^*)}{2}$. The deviation margin is $\epsilon = \frac{N_s (p^* - p)}{2} = \frac{N_s \Delta(q)}{2}$.
% Assuming $\Delta(q) > 0$, the probability of a non-neighbor crossing $t$ (False Positive) or a neighbor falling below $t$ (False Negative) is bounded by:
% \begin{equation}
% \Pr(\text{Error}) \le \exp\left(-\frac{2\epsilon^2}{N_s}\right) = \exp\left(-\frac{1}{2} N_s \Delta(q)^2\right).
% \end{equation}
% Substituting the condition for strong discriminability $\Delta(q) \ge \tau_\delta = \sqrt{\frac{2 \ln(1/\delta)}{N_s}}$:
% \begin{equation}
% \Pr(\text{Error}) \le \exp\left(-\frac{1}{2} N_s \cdot \frac{2 \ln(1/\delta)}{N_s}\right) = \exp(-\ln(1/\delta)) = \delta.
% \end{equation}
% Thus, the error is strictly controlled below $\delta$.
% If $\Delta(q) < \tau_\delta$ (weak discriminability), the gap is too small to satisfy the confidence level $\delta$.
% In the extreme case where $\Delta(q) \le 0$, the expected score of the signal is lower than or equal to the noise ($p^* \le p$). The midpoint logic collapses, and the "error" probability (overlap) essentially becomes the entire probability mass, approaching 1. Mathematically, as $\Delta(q) \to 0$, the bound $\exp(-\frac{1}{2} N_s \Delta(q)^2) \to e^0 = 1$.
\end{proof}
By operating directly on observable SC-score levels and scanning in descending order, Algorithm~\ref{candidate_selection} implements a query-driven stopping rule, where the decision to stop collecting candidates is governed solely by the available SC-score ordering and the refinement budget.
In summary, Algorithm~\ref{candidate_selection} is both theoretically justified by the SC-score separation properties and practically effective as it adapts to the intrinsic discriminability of each query without extra estimation overhead.

\begin{algorithm}[tb]
% \small
% \footnotesize
	\caption{$k$-ANNS Query}       
	\label{kann_query}
	\LinesNumbered
	\KwIn{Dataset $\mathcal D$, dataset size $n$, data dimensionality $d$, a query point $q$, number of results $k$, subspace number $N_s$, subspace dimensionality $s$, collision ratio $\alpha$, re-rank ratio $\beta$, number of K-means clusters $K$, the centroid list $centroids$, the IMI list $IMIs$}
	\KwOut{$k$ nearest points to $q$ in $\mathcal D$}
        % Initialize an array $SC\_scores$ of length $n$ and set to 0; \\
        Transform $q$ with the eigenvalues obtained from the \emph{Subspace-oriented Data Transformation} algorithm; \\
        Divide $q$ into $N_s$ subspaces: $q^1,\ldots,q^{N_s}$; \\
	\For{$i=1$ to $N_s$}{
            Further split $q^{i}$ into two parts $q^{i}_1$ and $q^{i}_2$; \\
            Calculate the distance between $q^{i}$ and centroids in each part, obtain $dists^{i}_1$, $dists^{i}_2$ and sorted indices $idx^{i}_1$, $idx^{i}_2$; \\
            $retrieved\_clusters \leftarrow$ \textbf{call} \emph{Scalable Dynamic Activation} $(\alpha,n,K,dists^{i}_1,idx^{i}_1,dists^{i}_2,idx^{i}_2,IMIs[i-1])$; \\
            $SC\_scores \leftarrow$ Derived from $retrieved\_clusters$; \\
	}
        $\mathcal C, candidate\_num \leftarrow $ \textbf{call} \emph{Query-aware Candidates Selection}$(\mathcal D, SC\_scores, \beta, N_s, n)$; \\
	\Return the refined \emph{top}-$k$ points closest to $q$ in $\mathcal C$; \\
\end{algorithm}

\subsection{$k$-ANNS Query}

Algorithm~\ref{kann_query} illustrates how TaCo supports \emph{$k$-ANNS} queries.
The procedure starts by transforming the query point $q$ using the eigenvalues obtained from the \emph{Subspace-oriented Data Transformation} algorithm, followed by partitioning the transformed query into $N_s$ subspaces (lines~1–2).
For each subspace, the distances between the query point and all centroids in both of its divided parts are calculated, and the indices corresponding to these distances sorted in ascending order are obtained (lines~3-5).
Subsequently, the \emph{Scalable Dynamic Activation} algorithm is invoked to retrieve clusters containing data points that collide with query $q$ in each subspace (line~6). 
The collision counts for all data points are then aggregated accordingly (line~7).
After processing all subspaces, the \emph{Query-aware Candidates Selection} algorithm is invoked to generate a candidate set based on the provided SC-score distribution (line~8). 
Finally, all candidates are refined by their exact distances to $q$, and the top-$k$ points with the smallest distances are returned (line~9).

Algorithm~\ref{kann_query} reflects the query pipeline under the subspace collision framework. 
Since both TaCo and SuCo are instances of this framework, their overall pipelines are structurally similar.
TaCo’s key innovation lies in the design of the internal modules within the pipeline (Algorithms~\ref{data_transformation}, ~\ref{scalable_dynamic_activation}, ~\ref{candidate_selection}), which leads to a significant performance gain over the query algorithm used in SuCo.

\section{Experimental Evaluation} \label{experimental_evaluation}

In this section, we evaluate the performance of TaCo through self-evaluations and comparative experiments against state-of-the-art ANNS methods, including both subspace collision-based and non-subspace collision-based methods. 
All methods are implemented in C/C++ and parallelized using OpenMP and/or Pthreads, with SIMD instructions utilized to accelerate computation. 
TaCo is developed in C++ and compiled with -O3 optimization.
% Our implementation of TaCo is developed in C++ and compiled with -O3 optimization, leveraging OpenMP for parallelization and SIMD for accelerating calculations.
All experiments are performed on a server with two AMD EPYC 9554 CPUs operating at 3.10 GHz and 756 GB of memory, running Ubuntu 22.04.

\subsection{Experimental Setup}

\noindent \textbf{Datasets and Queries.} 
We use five public real-world datasets with varying sizes and dimensionalities, which are widely adopted in ANNS studies~\cite{detlsh,wei2025subspace,malkov2018efficient,gao2024rabitq}.
The datasets include: 256-dimensional DEEP1M, 960-dimensional GIST1M, 128-dimensional SIFT10M, 96-dimensional Yandex DEEP10M, and 100-dimensional Microsoft SPACEV10M.
Note that SIFT10M, Yandex DEEP10M, and Microsoft SPACEV10M are randomly sampled subsets of their corresponding billion-scale datasets\footnote{http://corpus-texmex.irisa.fr; https://big-ann-benchmarks.com/neurips21.html}. 
For each dataset, we randomly select 100 points as query points and remove them from the original datasets.

\noindent \textbf{Benchmark Methods.}
We compare TaCo with eleven state-of-the-art in-memory ANN methods, covering both subspace collision-based and non-subspace collision-based methods.
(1) For subspace collision-based methods: \textbf{SuCo}~\cite{wei2025subspace} is the state-of-the-art method under the subspace collision framework before our work.
To evaluate the impact of the our proposed data transformation and query strategy optimizations on algorithm performance, we propose three ablation methods: \textbf{SuCo-DT} (SuCo with subspace-oriented data transformation), \textbf{SuCo-CS} (SuCo with query-aware candidates selection), and \textbf{SuCo-QS} (SuCo with all optimized query strategies).
\textbf{SC-Linear} serves as a baseline that performs linear scan without any index structure.
(2) For non-subspace collision-based methods: \textbf{IMI-OPQ}~\cite{ge2013optimized} uses the same IMI structure as TaCo and is recognized for its strong performance in the FAISS library~\cite{douze2024faiss}.
\textbf{DET-LSH}~\cite{detlsh} follows a query pipeline similar to TaCo, involving data projection, candidate selection, and final re-ranking.
\textbf{IVF-RaBitQ}~\cite{gao2025practical} is the state-of-the-art quantization method, achieves unbiased estimation and has been widely used in the industry.
\textbf{HNSW}~\cite{malkov2018efficient} is a widely adopted graph-based method known for its effectiveness in query processing.
\textbf{MIRAGE}~\cite{voruganti2025mirage} and \textbf{SHG}~\cite{gong2025accelerating} are optimized based on HNSW and demonstrate strong performance.

\noindent \textbf{Evaluation Measures.}
We evaluate all methods using six performance measures: indexing time, index memory footprint, query time, queries per second (QPS), recall, and mean relative error (MRE)~\cite{patella2008many,patella2009approximate,aumuller2020ann}. 
Among these, indexing time, query time, and QPS reflect efficiency; index memory footprint indicates storage overhead; and recall and MRE assess the quality of the returned results.
Note that offline data preprocessing steps (e.g., data projection for DET-LSH, data rotation for IMI-OPQ, and data transformation for TaCo) are not included in the indexing time, which is a common practice in previous works for evaluating experiments~\cite{li2025saq,detlsh,dblsh,ge2013optimized}.
For a query $q$, let $\mathcal R=\{o_1,\ldots,o_k\}$ be the returned result set and $\mathcal R^*=\{o_1^*,\ldots,o_k^*\}$ be the exact $k$-NN set. Recall is defined as $\frac{\lvert \mathcal R \cap \mathcal R^* \rvert}{k}$, and mean relative error (MRE) is defined as $\frac{1}{k} \sum_{i=1}^{k} \frac{\left\|q,o_i\right\|-\left\|q,o_i^*\right\|}{\left\|q,o_i^*\right\|}$.

\noindent \textbf{Parameter Settings.} 
Both the indexing and query phases of ANNS methods involve configurable parameters. 
For indexing parameters, we adhere to the optimal configurations recommended in previous studies~\cite{echihabi2018lernaean,hydra2,li2019approximate,detlsh,wei2025subspace,wang2021comprehensive} and empirically validate these configurations to ensure all methods achieve optimal indexing performance. 
For query parameters, we dynamically adjust their values to examine the trade-offs between query efficiency and accuracy, enabling a comprehensive evaluation across different settings.
For TaCo, SuCo, SuCo-DT, SuCo-CS, and SuCo-QS, $N_s\in [4,10]$, $s\in [6,12]$, $\alpha \in \left[0.01, 0.1\right]$, $\beta \in \left[0.001, 0.05\right]$.
For IMI-OPQ, $M=2$, $K=2^{18}$, $\beta \in \left[10^{-4},10^{-1}\right]$.
For DET-LSH, $L=4$, $K=16$, $\beta \in \left[0.005,0.2\right]$, $c=1.5$.
For IVF-RaBitQ, we adopt the default settings of the RaBitQ library, i.e., $K = 2^{12}$ and $B = 3$.
For HNSW, $efConstruction=200$, $M=25$, $efSearch \in \left[300,3000\right]$.
For MIRAGE, $S=32$, $R=4$, $iter=15$.
For SHG, $M=48$, $efConstruction=80$.
The value of $k$ in $k$-ANNS is set to 50 by default.

\begin{table}
% \small
	\centering
	\caption{Comparison on TaCo and SC-Linear}
  % \vspace*{-0.2cm}
	\label{scevaluation}
\begin{tabular}{@{}ccccc@{}}
\toprule
\textbf{Method}                               & \textbf{Dataset} & \textbf{Query Time (ms)} & \textbf{Speedup} & \textbf{Recall} \\ \midrule
\multirow{2}{*}{SC-Linear}                    & DEEP1M          & 273.7                  & \textbackslash{} & 0.968           \\
                                              & SIFT10M         & 4786.1                  & \textbackslash{} & 0.9906           \\ 
\multicolumn{1}{c}{\multirow{2}{*}{TaCo}} & DEEP1M          & 1.266                    & 216.2            & 0.9358          \\
\multicolumn{1}{c}{}                        & SIFT10M         & 6.705                   & 713.8            & 0.9726          \\ \bottomrule
\end{tabular}
% \vspace*{-0.2cm}
\end{table}

\begin{table}
% \small
% \footnotesize
	\centering
	\caption{Parameter settings for each dataset}
  % \vspace*{-0.1cm}
	\label{parameter_settings}
	\begin{tabular}{ccccc}
		\toprule
		\textbf{Dataset} & \textbf{Dim.} & \textbf{$N_s$} & \textbf{$s$} & \textbf{Dim. Reduction}\\
		\midrule
		DEEP1M & 256  & 6 & 8 & 81.25\%\\
        GIST1M & 960  & 4 & 10 & 95.83\%\\
		SIFT10M & 128  & 6 & 6 & 71.88\%\\
        Yandex Deep10M & 96 & 6 & 8 & 50\%\\
        Microsoft SPACEV10M & 100 & 6 & 10 & 40\%\\
		\bottomrule
	\end{tabular}
 % \vspace*{-0.2cm}
\end{table}

\subsection{Self-evaluation of TaCo}

\subsubsection{TaCo vs. SC-Linear}
We first compare the query performance between TaCo and SC-Linear under same parameter settings ($\alpha = 0.05$, $\beta = 0.005$). 
As shown in Table~\ref{scevaluation}, TaCo achieves a significant improvement in query efficiency: exceeding $\mathbf{700\times}$ speedup on the SIFT10M dataset, with only a marginal decrease in recall. 
The speedup ratio increases substantially with the size of the dataset, demonstrating that TaCo has strong scalability.
These results confirm that TaCo's index structure and query strategies can be effectively applied to the subspace collision framework for ANNS.

\subsubsection{Scalable Dynamic Activation vs. Dynamic Activation} \label{imiquery_compare}
The \emph{Scalable Dynamic Activation} algorithm, introduced in Section~\ref{collision_counting}, addresses scalability issues in IMI collision counting.
Figure~\ref{imi_query} compares the efficiency of the \emph{Scalable Dynamic Activation} and the original \emph{Dynamic Activation} algorithms on the SIFT10M dataset under varying numbers of clusters $K$ and collision ratios $\alpha$.
When $K$ is small, both algorithms perform similarly, with the \emph{Dynamic Activation} algorithm even slightly outperforming in some cases. 
However, as $K$ increases, the \emph{Scalable Dynamic Activation} algorithm exhibits a notable performance improvement, achieving up to a 30\% reduction in query time. 
This gain comes from its min-heap structure, which reduces query complexity to $\mathcal{O}(1)$ versus the original's $\mathcal{O}(\sqrt{K})$. 
For small $K$, heap operations (push and pop) slightly reduce performance due to update overhead. 
As $K$ increases, the linear query cost of \emph{Dynamic Activation} becomes a bottleneck, while the constant-time queries of \emph{Scalable Dynamic Activation} provide significant advantages, confirming its superior scalability.

\begin{figure}[tb] 
% \vspace*{-0.3cm}
	\centering
	\includegraphics[width=0.6\linewidth]{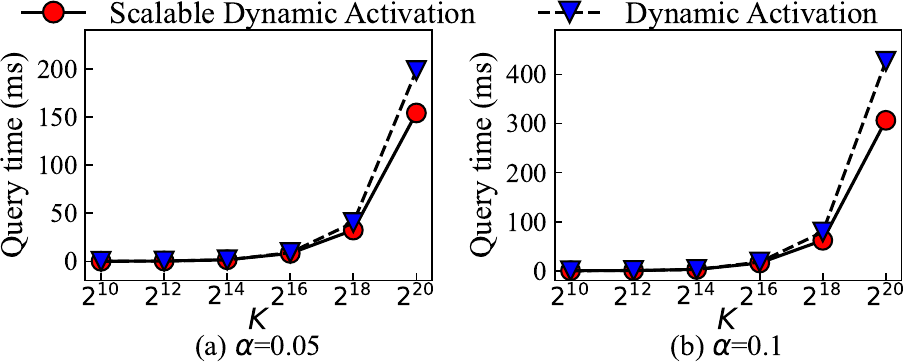}
 % \vspace*{-0.2cm}
	\caption{Efficiency comparison between Scalable Dynamic Activation and Dynamic Activation algorithms on SIFT10M.}
	\label{imi_query}
  % \vspace*{-0.2cm}
\end{figure}

\begin{figure*}[tb] 
% \vspace*{-0.3cm}
	\centering
	\includegraphics[width=\linewidth]{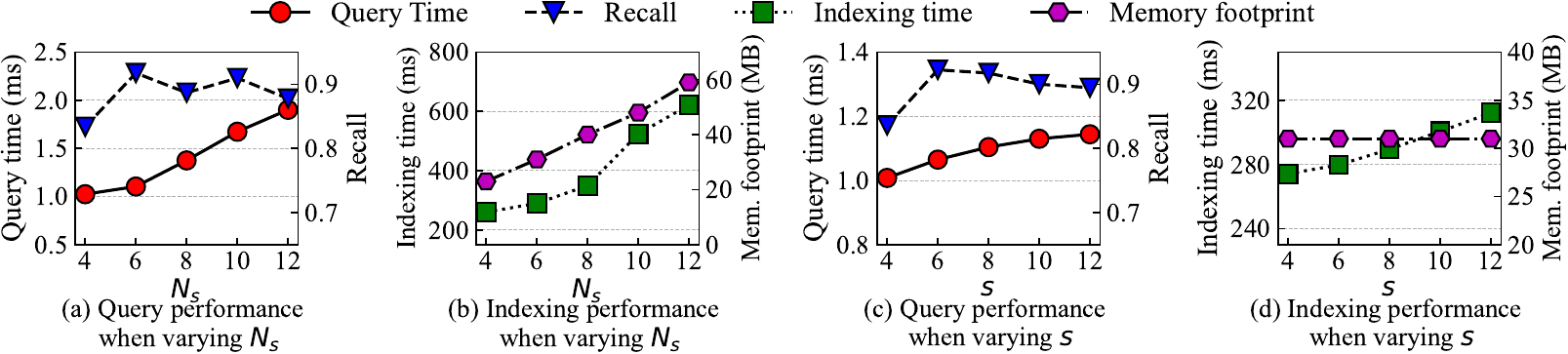}
 % \vspace*{-0.2cm}
	\caption{Performance of TaCo when varying the number of subspaces $N_s$ and the subspace dimensionality $s$ on DEEP1M.}
	\label{parameter_study}
 % \vspace*{-0.1cm}
\end{figure*}

\subsubsection{Parameter study}

Parameters $N_s$ and $s$ critically affect the indexing and query performance of TaCo.
As shown in Figure~\ref{parameter_study} (a) and (b), increasing $N_s$ generally leads to higher indexing time and memory consumption. 
While query performance improves initially with $N_s$ and reaches an optimum at $N_s=6$, beyond which further increasing $N_s$ leads to a gradual rise in query time and may even reduce recall.
As shown in Figure~\ref{parameter_study} (c) and (d), increasing $s$ leads to higher computational costs, consequently resulting in increased indexing and query time.
% In contrast to $N_s$, the memory footprint is less sensitive to changes in $s$.
Furthermore, recall reaches its optimal value at $s=6$.
Table~\ref{parameter_settings} summarizes the optimal parameter settings of TaCo for each dataset, where $1-\frac{N_s \cdot s}{d}$ represents the dimensionality reduction rate.
It can be observed that the dimensionality reduction rates achieved across datasets are impressive, ranging from 40\% to as high as 95.83\%.
The dimensionality reduction is a key factor contributing to TaCo's superior indexing and query efficiency compared to subspace collision-based methods.

\subsection{TaCo vs. Subspace Collision Methods} \label{taco_vs_sc}

% In this section, we compare TaCo with three subspace collision-based methods: SuCo, SuCo-DT, and SuCo-QS.
% SuCo is the state-of-the-art subspace collision-based method.
% To evaluate the impact of the our proposed data transformation and query strategy optimizations, we conduct an \emph{ablation study} by integrating each of these two optimizations separately with SuCo, resulting in two ablation variants: SuCo-DT (SuCo with subspace-oriented data transformation) and SuCo-QS (SuCo with the optimized query strategy).

% This section evaluates TaCo against four subspace collision-based methods.
% We employ SuCo, the current state-of-the-art subspace collision method, as our baseline.
% To systematically evaluate our proposed data transformation and query strategy optimizations, we conduct an \emph{ablation study} by integrating each component separately with SuCo, yielding three ablation variants: SuCo-DT (with data transformation), SuCo-CS (with query-aware candidates selection) and SuCo-QS (with all optimized query strategy).

\subsubsection{Indexing performance}

Figure~\ref{indexing_sc} shows the indexing time and index memory footprint of all subspace collision-based methods.
Note that the indexing time does not include preprocessing time.
Owing to the use of subspace-oriented data transformation in both TaCo and SuCo-DT, these two methods achieve same indexing performance. 
Similarly, SuCo-QS, SuCo-CS, and SuCo, which do not incorporate this transformation, also exhibit same indexing performance.
Therefore, we focus our analysis on comparing TaCo and SuCo.
Experimental results indicate that TaCo achieves up to $\mathbf{8\times}$ speedups in indexing while consuming only $\mathbf{0.6\times}$ memory footprint when compared with SuCo. 
This significant improvement can be attributed to three main factors: dimensionality reduction, fewer subspaces, and code-level optimizations. 
These results collectively confirm the effectiveness of the proposed subspace-oriented data transformation in enhancing indexing performance.

\subsubsection{Query performance}

In $k$-ANNS tasks, a fundamental trade-off exists between query efficiency and accuracy: longer query processing times generally yield higher-quality results, whereas an excessive focus on efficiency may compromise result quality. Therefore, a comprehensive evaluation of query performance must consider both efficiency (measured by QPS) and accuracy (measured by recall and MRE).
Figure~\ref{recall_qps_mre_sc} presents the recall-QPS and MRE-QPS curves for all subspace collision-based methods. 
The results show that TaCo achieves over $\mathbf{1.5\times}$ QPS at 0.9 recall when compared with SuCo.
Furthermore, SuCo-DT, SuCo-CS, and SuCo-QS exhibit better query performance than the original SuCo method.
These ablation experiments confirm that the three optimization techniques we proposed can effectively improve query performance.

\begin{figure}[tb] 
% \vspace*{-0.2cm}
	\centering
	\includegraphics[width=0.9\linewidth]{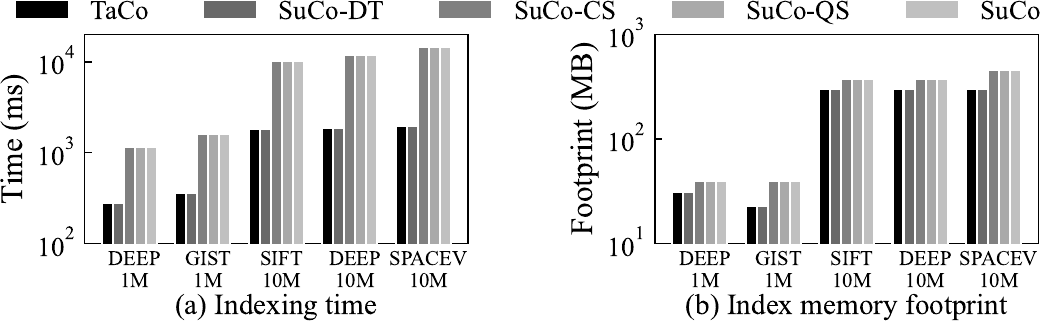}
	\caption{Indexing performance comparison between TaCo and subspace collision-based methods.}
	\label{indexing_sc}
 % \vspace*{-0.1cm}
\end{figure}

\begin{figure*}[tb] 
% \vspace*{-0.3cm}
	\centering
	\includegraphics[width=\linewidth]{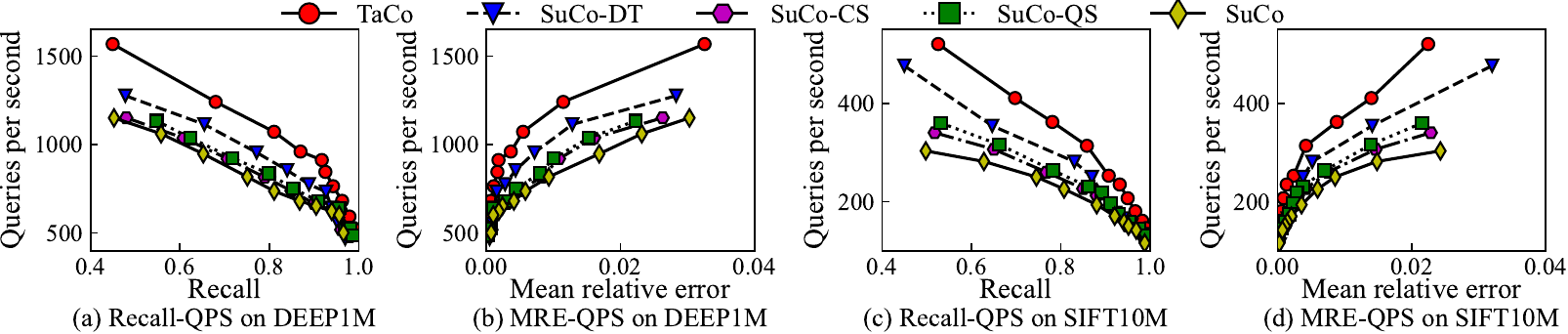}
 % \vspace*{-0.3cm}
	\caption{Query performance comparison between TaCo and subspace collision-based methods.}
	\label{recall_qps_mre_sc}
 % \vspace*{-0.1cm}
\end{figure*}

\begin{figure*}[tb] 
	\centering
	\includegraphics[width=\linewidth]{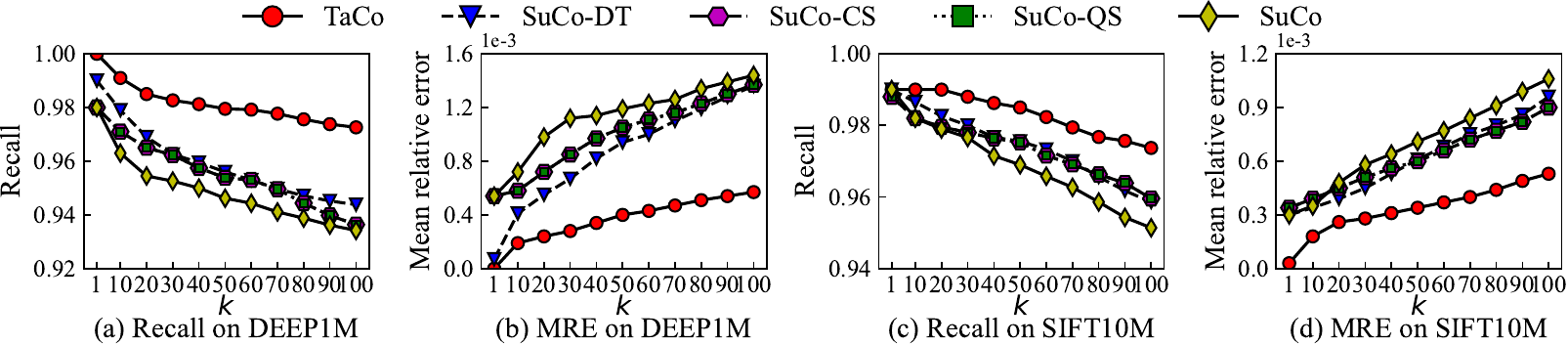}
	\caption{Performance under different $k$: TaCo vs. subspace collision-based methods.}
	\label{diffk}
\end{figure*}

\subsubsection{Performance under different $k$}
% In $k$-ANNS tasks, the parameter $k$ denotes the number of nearest neighbors to be retrieved. 
% A larger $k$ indicates a more difficult query task. 
The ability of a $k$-ANNS method to maintain stable performance across varying values of $k$ is an important aspect of its scalability.
Figure~\ref{diffk} illustrates the query performance of all subspace collision-based methods under different $k$ values ranging from 1 to 100.
SuCo-CS and SuCo-QS overlap in the figure because, under the same parameter settings, the Scalable Dynamic Activation and the Dynamic Activation algorithms return identical results and therefore do not affect recall.
As $k$ increases, all methods exhibit a slight decline in query accuracy. 
This is primarily because the number of candidate points remains fixed irrespective of $k$, a larger $k$ thus increases the probability of missing true nearest neighbors. 
Nevertheless, TaCo consistently outperforms all other subspace collision-based methods across all evaluated settings of $k$, demonstrating its robustness and superior scalability.

\subsection{TaCo vs. Non-Subspace Collision methods}

% In this section, we compare TaCo with three non-subspace collision-based methods: IMI-OPQ (VQ-based), DET-LSH (LSH-based), and HNSW (graph-based).
% Different categories of ANNS methods exhibit distinct strengths and weaknesses in indexing and query performance, which makes direct comparisons challenging. Nevertheless, we provide a comparative analysis between TaCo and other categories of ANNS methods to provide readers with an intuitive and comprehensive perspective.

This section presents a comparative analysis between TaCo and six representative methods from different ANNS paradigms: IMI-OPQ and IVF-RaBitQ (VQ-based); DET-LSH (LSH-based); HNSW, MIRAGE, and SHG (graph-based). 
Although cross-paradigm comparisons are non-intuitive due to fundamental architectural differences, this evaluation offers valuable insights into TaCo's competitive positioning relative to non-subspace collision methods.

\begin{figure*}[tb] 
	\centering
	\includegraphics[width=0.9\linewidth]{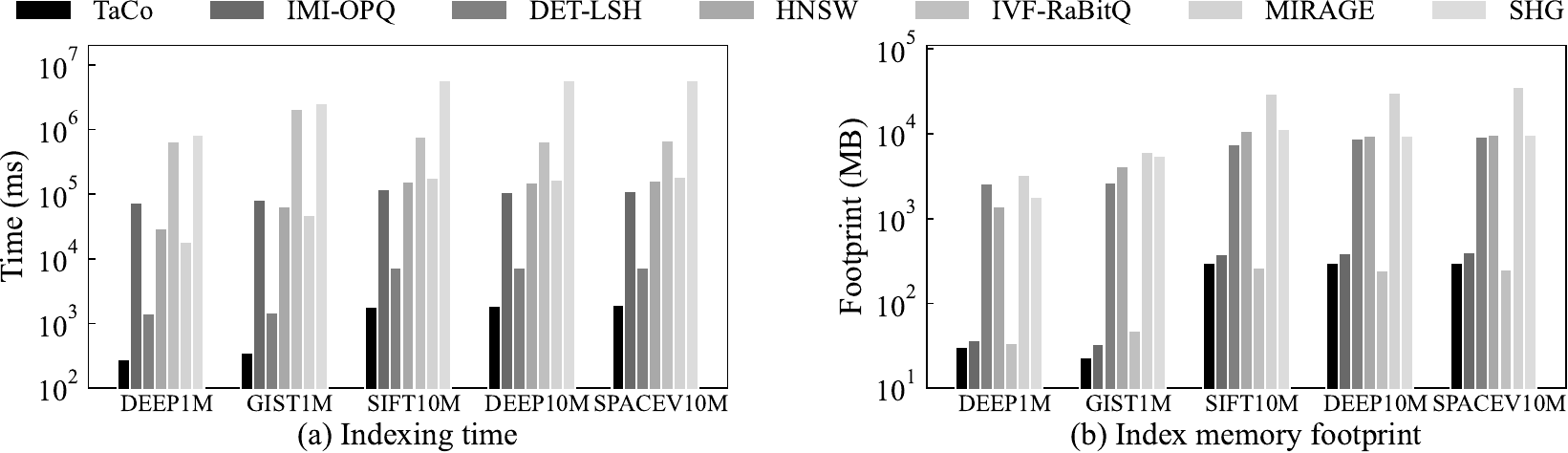}
	\caption{Indexing performance comparison between TaCo and non-subspace collision-based methods.}
	\label{indexing_nonsc}
\end{figure*}

\begin{figure*}[tb] 
% \vspace*{-0.3cm}
	\centering`
	\includegraphics[width=\linewidth]{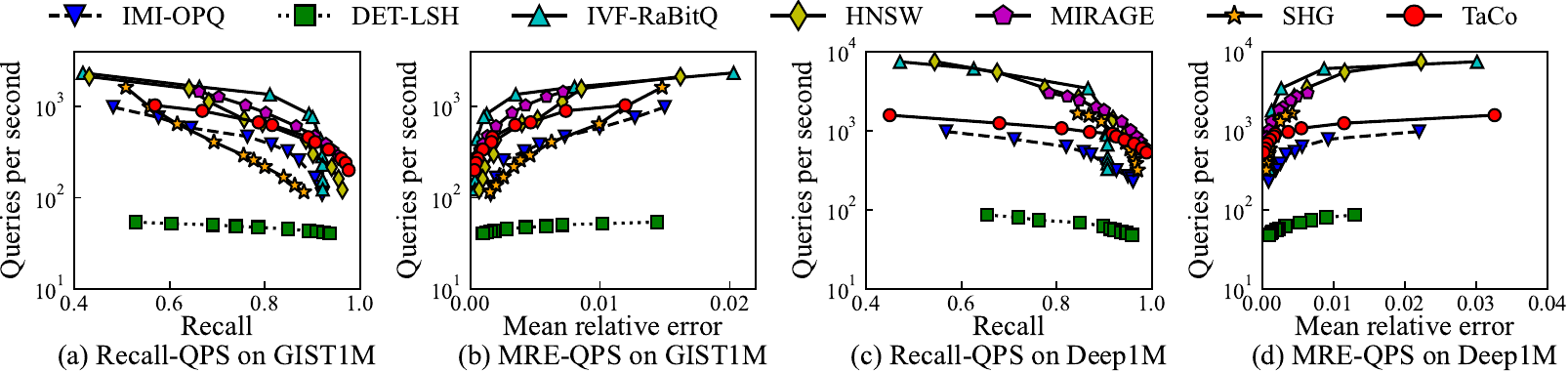}
 % \vspace*{-0.3cm}
	\caption{Query performance comparison between TaCo and non-subspace collision-based methods.}
	\label{recall_qps_mre_nonsc}
 % \vspace*{-0.1cm}
\end{figure*}

\subsubsection{Indexing performance} \label{index_nonsc}

Figure~\ref{indexing_nonsc}(a) shows that TaCo achieves the best indexing efficiency, while Figure~\ref{indexing_nonsc}(b) shows that TaCo, IMI-OPQ, and IVF-RaBitQ achieve the lowest memory consumption among all evaluated methods.
Note that the reported indexing time excludes preprocessing, and the index memory footprint does not include the storage cost of the dataset.
TaCo, IMI-OPQ, and IVF-RaBitQ utilize the inverted file-based structure (IMI and IVF), whose lightweight design lead to significantly lower memory usage compared to DET-LSH (tree-based index) and HNSW, MIRAGE, SHG (graph-based index). 
However, due to differences in design objectives, TaCo performs coarse-grained collision probing within each subspace under the subspace collision framework, while IMI-OPQ and IVF-RaBitQ require fine-grained partitioning of all data points in the original space. 
As a result, compared with IMI-OPQ and IVF-RaBitQ, TaCo reduces indexing time by 2\textasciitilde3 orders of magnitude.
In contrast, DET-LSH relies on fine-grained tree-based partitioning, and HNSW, MIRAGE, SHG connect each point to its nearest neighbors in a graph structure. 
These approaches result in heavier index structures, requiring longer indexing times and larger memory footprints.
In summary, TaCo achieves superior indexing performance due to its efficient subspace-oriented design and lightweight index.

\subsubsection{Query performance}
Figure~\ref{recall_qps_mre_nonsc} presents the recall–QPS and MRE–QPS curves for TaCo and non-subspace collision-based methods.
Experimental results indicate that IVF-RaBitQ, HNSW, and MIRAGE perform well at moderate recall levels (0.6\textasciitilde0.9), whereas TaCo exhibits better stability at higher recall levels (>0.9), particularly on the hard datasets such as GIST1M.
Graph-based methods (HNSW, MIRAGE, SHG) achieve high query efficiency at the cost of substantial indexing time and memory consumption.
During indexing, they identify and connect near neighbors for each data point, while queries follow a gradually converging path through the graph. 
These methods perform better on large and simple datasets like Deep1M. 
However, on hard datasets such as GIST1M, their greedy search strategy is prone to becoming trapped in local optima, limiting query performance.
In contrast, TaCo employs a subspace collision probing strategy that offers greater robustness and fault tolerance. 
This allows it to perform consistently well even on hard datasets. 
IVF-RaBitQ leverages SIMD operations to accelerate computations on quantized codes, thus achieving QPS similar to graph-based methods; however, the limited number of quantization bits makes it challenging to attain very high recall.
In summary, TaCo achieves state-of-the-art indexing performance, while simultaneously attaining top-tier query performance, establishing it as a lightweight and effective solution for $k$-ANNS tasks.

\begin{figure*}[tb] 
	\centering
	\includegraphics[width=\linewidth]{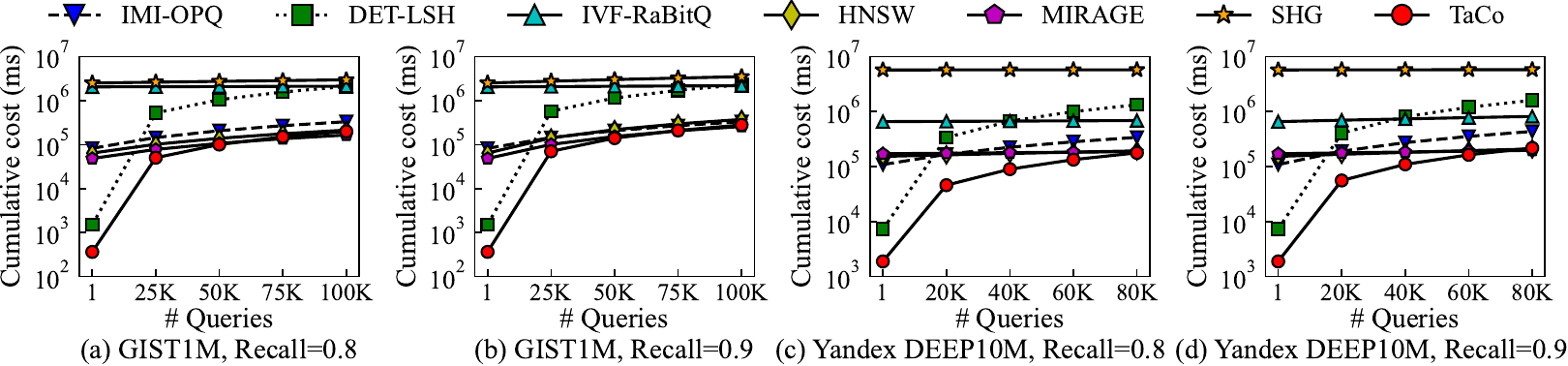}
 % \vspace*{-0.3cm}
	\caption{Cumulative query cost (start with the indexing time): TaCo vs. non-subspace collision-based methods.}
	\label{query_indexing}
 % \vspace*{-0.1cm}
\end{figure*}

\subsubsection{Overall Evaluation}

Given the considerable differences in indexing and query performance across different categories of ANNS methods, it is essential to conduct a comprehensive evaluation that integrates both indexing and query metrics.
Figure~\ref{query_indexing} shows the cumulative query costs of all methods under the same recall, where the cost starts with the indexing time.
The experimental results demonstrate the advantages of TaCo: it completes index construction and processes 50K-80K queries before the best competitor (i.e., HNSW and MIRAGE) answers its first query.
In contrast, the slow index construction of IVF-RaBitQ and SHG, as well as the slow query time of DET-LSH, make them less competitive in this evaluation. 
In summary, while different categories of methods are suited to particular scenarios due to their distinctive characteristics, TaCo achieves an effective balance between indexing and query efficiency, making it especially suitable for applications that require rapid deployment and immediate query responses, such as retrieval-based sparse attention for LLM inference acceleration~\cite{zhang2025pqcache,liu2024retrievalattentionacceleratinglongcontextllm}.

\section{Related Work} \label{related_work}

\textbf{Approximate Nearest Neighbor Search.} 
With data being generated at an unprecedented rate and datasets increasing in scale~\cite{zhuang2024cbcms,fu2024securing,fu2023ti,fei2023flexauth,zhang2025polaris,zhuang2026structured}, more efficient management of large-scale data is required to facilitate advanced analysis~\cite{DBLP:journals/sigmod/Palpanas15,Palpanas2019,li2024disauth,wei2023data,wei2025dominate}.
A wide range of effective ANNS methods have been proposed~\cite{zeyubulletin-sep23,DBLP:journals/debu/0007X0H0P024}: locality-sensitive hashing (LSH)-based methods~\cite{dblsh,lccslsh,pmlsh,andoni2015optimal,detlsh}, vector quantization (VQ)-based methods~\cite{jegou2010product,ge2013optimized,norouzi2013cartesian,babenko2014inverted,gao2024rabitq,gao2025practical}, tree-based methods~\cite{dpisaxjournal,annoy,coconut,messi,dumpy,wang2024dumpyos,seanetconf,leafi}, graph-based methods~\cite{li2019approximate,fu2019fast,malkov2018efficient,voruganti2025mirage,azizi2025graph,gong2025accelerating}, and hybrid methods~\cite{chen2018sptag,lshapg,azizi2023elpis,gou2025symphonyqg}. 
% Several comprehensive surveys are available to help readers gain a deeper understanding of the trade-offs and comparisons among different ANNS algorithms~\cite{echihabi2018lernaean,hydra2,azizi2025graph,zeyubulletin-sep23,li2019approximate,annbulletin,wang2021comprehensive}.
Subspace collision is a newly proposed framework, which achieves a better balance between index construction and query answering~\cite{wei2025subspace}.
Our work advances the subspace collision framework by introducing subspace-oriented data transformation and optimizing query strategies, thereby extending the boundaries of efficient and scalable ANNS.

\noindent \textbf{AI for ANNS and ANNS for AI.}
Traditional ANNS methods rely on heuristic designs, limiting their ability to adapt indexing and query strategies to data distributions or query characteristics.
In contrast, deep learning can capture intrinsic features that are difficult to manually design~\cite{lecun2015deep}, motivating a range of learning-enhanced ANNS approaches, including learn-to-index~\cite{zeng2025lira,gupta2022bliss,li2023learning,chiu2019learning}, learn-to-query~\cite{baranchuk2019learning,feng2023reinforcement}, learned early termination~\cite{li2020improving,chatzakis2025darth}, and learned cardinality estimation~\cite{sun2021learned,zheng2023learned}.
% Traditional ANNS methods are typically designed based on heuristics principles, which limits their ability to dynamically adapt indexing and query strategies according to data distribution or query characteristics.
% Deep learning techniques can capture intrinsic data features that are difficult for humans to interpret~\cite{lecun2015deep}. 
% As a result, they have been applied to tackle the above challenges in ANNS~\cite{li2019approximate}, giving rise to learning-enhanced approaches such as learn-to-index~\cite{zeng2025lira,gupta2022bliss,li2023learning,chiu2019learning}, learn-to-query~\cite{baranchuk2019learning,feng2023reinforcement}, learned early termination~\cite{li2020improving,chatzakis2025darth}, and learned cardinality estimation~\cite{sun2021learned,zheng2023learned}.
Furthermore, with advances in large language models (LLMs)~\cite{zhao2023survey}, ANNS has gained significant attention for enhancing AI applications~\cite{wei2026virtuouscycleaipoweredvector}. 
For instance, retrieval-augmented generation (RAG) uses ANNS to provide contextual information for LLM queries~\cite{lewis2020retrieval,fan2024survey,gao2023retrieval}, and key-value cache (KVCache) employs ANNS to accelerate model inference~\cite{zhang2025pqcache,liu2024retrievalattentionacceleratinglongcontextllm,desaihashattention,song2026csattention}.

\noindent \textbf{Hybrid Search.}
Traditional ANNS primarily focuses on vector search. 
However, how to achieve hybrid search by integrating ANNS with traditional database systems has become a key concern in both academia~\cite{gollapudi2023filtered,patel2024acorn,zuo2024serf,xu2024irangegraph} and industry~\cite{yang2022oceanbase,wang2021milvus,mohoney2023high,wei2020analyticdb}.
% (including Oceanbase~\cite{yang2022oceanbase,yang2023oceanbase}, Zilliz~\cite{wang2021milvus}, Apple~\cite{mohoney2023high}, and Alibaba~\cite{wei2020analyticdb}).
Hybrid search enhances traditional ANNS by enabling vectors to satisfy additional attribute-based constraints~\cite{wang2023efficient,wu2022hqann}, greatly expanding its applicability in real-world applications. 
For example, in e-commerce, users can specify attributes such as price range and product style while still leveraging semantic similarity search~\cite{cai2024navigating,xu2024irangegraph}.
In search engines, results can be filtered by domain or user privileges without compromising retrieval relevance~\cite{cai2024navigating}.
% Effectively leveraging this metadata for filtered search can greatly enhance retrieval accuracy and the overall performance of the RAG pipeline, thereby boosting the experience of generative AI applications.
In RAG systems, vector databases exploit metadata-based filtering to improve retrieval accuracy, thereby enhancing generative AI applications~\cite{fan2024survey,gao2023retrieval}.
This integration of vector-based retrieval and structured data filtering offers a practical paradigm for next-generation information retrieval systems.

\section{Conclusions}

In this paper, we first designed the subspace-oriented data transformation mechanism to enable data-adaptive subspace collision framework.
Next, we presented query-aware and scalable query strategies that dynamically allocate overhead for each query and accelerate collision probing within subspaces.
Then, we proposed the TaCo method, which achieves efficient and accurate ANN search while maintaining an excellent balance between indexing and query performance.
Finally, we conducted extensive experiments, and the results demonstrate the superiority of TaCo.
In future work, we will combine deep learning to explore learn-based index structures and query strategies under the subspace collision framework, as well as randomized numerical linear algebra techniques~\cite{derezinski2021Sparse,niu2025Fundamental}.

\begin{acks}
T.~Palpanas supported by EU Horizon projects TwinODIS ($101160009$) and DataGEMS ($101188416$), and by $Y \Pi AI \Theta A$ \& NextGenerationEU project HARSH ($Y\Pi 3TA-0560901$) that is carried out within the framework of the National Recovery and Resilience Plan “Greece 2.0” with funding from the European Union – NextGenerationEU. 
Z.~Liao would like to acknowledge the National Natural Science Foundation of China (via fund NSFC-12571561) and the Fundamental Research Support Program of HUST (2025BRSXB0004) for providing partial support.
\end{acks}

\bibliographystyle{ACM-Reference-Format}
\bibliography{ref}

\end{document}